\documentclass[11pt]{article}%
\usepackage{amsfonts}
\usepackage{amsmath}
\usepackage{fullpage}
\usepackage{amssymb}
\usepackage{hyperref}
\usepackage{graphicx}%
\providecommand{\U}[1]{\protect\rule{.1in}{.1in}}
\newtheorem{theorem}{Theorem}

\newtheorem{corollary}[theorem]{Corollary}

\newtheorem{definition}[theorem]{Definition}

\newtheorem{remark}[theorem]{Remark}

\newenvironment{proof}[1][Proof]{\noindent\textbf{#1.} }{\ \rule{0.5em}{0.5em}}
\numberwithin{equation}{section}
\begin{document}

\title{\textbf{Recoverability in quantum information theory}}
\author{Mark M. Wilde\thanks{Hearne Institute for Theoretical Physics, Department of
Physics and Astronomy, Center for Computation and Technology, Louisiana State
University, Baton Rouge, Louisiana 70803, USA}}
\maketitle

\begin{abstract}
The fact that the quantum relative entropy is non-increasing with respect to
quantum physical evolutions lies at the core of many optimality theorems in
quantum information theory and has applications in other areas of physics. In
this work, we establish improvements of this entropy inequality in the form of
physically meaningful remainder terms. One of the main results can be
summarized informally as follows: if the decrease in quantum relative entropy
between two quantum states after a quantum physical evolution is relatively
small, then it is possible to perform a recovery operation, such that one can
perfectly recover one state while approximately recovering the other. This can
be interpreted as quantifying how well one can reverse a quantum physical
evolution. Our proof method is elementary, relying on the method of complex
interpolation, basic linear algebra, and the recently introduced R\'{e}nyi
generalization of a relative entropy difference.\ The theorem has a number of
applications in quantum information theory, which have to do with providing
physically meaningful improvements to many known entropy inequalities.

\end{abstract}

\section{Introduction}

Entropy inequalities are foundational in quantum information theory
\cite{NC10,W13}, giving limitations not only on which kinds of physical
evolutions are possible in principle but also on efficiencies of communication
tasks. More generally and for similar reasons, these inequalities find
application in many areas of physics such as thermodynamics
\cite{Brandao17032015}, condensed matter \cite{fradkin2013field}, and black
hole physics \cite{AMPS13} to name a few. The most prominent entropy
inequalities are the non-increase of quantum relative entropy with respect to
the application of a quantum channel \cite{Lindblad1975,U77}\ and the strong
subadditivity of quantum entropy \cite{PhysRevLett.30.434,LR73}. In fact,
these inequalities are known to be equivalent to each other.

A recent line of research, which has in part been motivated by the posting
\cite{Winterconj}, has been to establish physically meaningful refinements of
these entropy inequalities. For example, one might think that if the amount by
which the quantum relative entropy decreases is not very much, then it might
be possible to perform a recovery channel to reverse the action of the
original one. In fact, one of the earliest results in this spirit is due to
Petz \cite{Petz1986,Petz1988}, who showed that perfect reversal of a channel
acting on two given states is possible if and only if the relative entropy
decrease is equal to zero. Furthermore, he gave an explicit construction of
the recovery channel which does so (now called the \textit{Petz recovery
map}), such that it depends on the original channel and one of the states used
to evaluate the quantum relative entropy. These results were later extended in
order to elucidate the structure of quantum states that saturate the strong
subadditivity inequality \cite{HJPW04} and the structure of states and
channels which saturate the non-increase of quantum relative entropy
inequality \cite{Mosonyi2004,M05}.

The \textquotedblleft perfect saturation\textquotedblright\ results quoted
above are interesting from a fundamental perspective but seem to have little
bearing in applications. That is, one might wonder if the results still hold
in some form when the entropy inequalities are not fully saturated but are
instead nearly saturated. After an initial negative result in this direction
\cite{ILW08}, a breakthrough result \cite{FR14}\ established a long desired
refinement of the strong subadditivity inequality. In particular, the new
contribution showed that if the strong subadditivity inequality is nearly
saturated, then the relevant tripartite state is an approximate quantum Markov
chain, in the sense that it is possible to recover one system by acting
exclusively on one other system while at the same time preserving the
correlations with a third. Later work has further elucidated the form of the
recovery channel used in approximate quantum Markov chains \cite{SOR15}. The
result from \cite{FR14} has now found a number of applications in quantum
information theory \cite{LW14,SBW14,SW14,W14}\ and is expected to find more in
other areas of physics.

The main contribution of the present paper is to establish physically
meaningful refinements of the non-increase of quantum relative entropy with
respect to quantum channels. One of the main results can be summarized
informally as follows: if the decrease in quantum relative entropy between two
quantum states after a quantum channel acts is relatively small, then it is
possible to perform a recovery operation, such that one can perfectly recover
one state while approximately recovering the other. A significant advantage of
the proof detailed here is that it is elementary, relying on standard methods
from the theory of complex interpolation \cite{BL76,RS75}, basic linear
algebra, and the notion of a R\'{e}nyi generalization of a relative entropy
difference \cite{SBW14}. The refinement of strong subadditivity from
\cite{FR14}\ is now a corollary of Theorem~\ref{thm:rel-ent} presented here,
but it remains open to determine whether the converse implication is true or
whether the more general refinement of strong subadditivity from \cite{SOR15}
can be obtained from Theorem~\ref{thm:rel-ent}. Furthermore, the recovery
channel given here obeys desirable \textquotedblleft
functoriality\textquotedblright\ properties discussed in \cite{LW14}, which
allows for Theorem~\ref{thm:rel-ent} to be applied in a wide variety of contexts.

We begin in the next section with some brief background material and a
statement of the operator Hadamard three-line theorem.
Section~\ref{sec:main-result}\ details our main result
(Theorem~\ref{thm:rel-ent}) and Section~\ref{sec:func}\ details the
functoriality properties of the recovery channel presented here.
Section~\ref{sec:corollaries}\ shows how many refinements of entropy
inequalities follow as corollaries of Theorem~\ref{thm:rel-ent}. We conclude
in Section~\ref{sec:discuss}\ with a discussion and some open questions.

\section{Background}

For more background on quantum information theory, we refer to the books
\cite{NC10,W13}. Throughout the paper, we deal with density operators and
quantum channels. We restrict our developments to finite-dimensional Hilbert
spaces, even though it should be possible to extend some of the results here
to separable Hilbert spaces. (We leave this for future developments.) Density
operators are positive semi-definite operators with trace equal to one---they
represent the state of a quantum system. Quantum channels are linear
completely positive trace-preserving maps taking density operators in one
quantum system to those in another.
%With the aim of focusing on conceptual
%content and ease of exposition, we state all of our results for positive
%definite density operators. Here again it should be possible to extend the
%results more generally to all density operators by employing appropriate
%continuity arguments, but we leave that for future developments as well.

An important technical tool in this work is the Schatten $p$-norm of an
operator $A$, defined as%
\begin{equation}
\left\Vert A\right\Vert _{p}\equiv\left[  \operatorname{Tr}\left\{  \left\vert
A\right\vert ^{p}\right\}  \right]  ^{1/p},
\end{equation}
where $|A|\equiv\sqrt{A^{\dag}A}$ and $p\geq1$. The convention is for
$\left\Vert A\right\Vert _{\infty}$ to be defined as the largest singular
value of $A$ because $\left\Vert A\right\Vert _{p}$ converges to this in the
limit as $p\rightarrow\infty$. In the proof of our main result
(Theorem~\ref{thm:rel-ent}), we repeatedly use the fact that $\left\Vert
A\right\Vert _{p}$ is unitarily invariant. That is, $\left\Vert A\right\Vert
_{p}$ is invariant with respect to linear isometries, in the sense that%
\begin{equation}
\left\Vert A\right\Vert _{p}=\left\Vert UAV^{\dag}\right\Vert _{p},
\end{equation}
where $U$ and $V$ are linear isometries satisfying $U^{\dag}U=I$ and $V^{\dag
}V=I$. From these norms, one can define information measures relating quantum
states and channels, with the main one used here known as a R\'{e}nyi
generalization of a relative entropy difference \cite{SBW14}, recalled in the
next section. A special case of this is the R\'{e}nyi conditional mutual
information defined in \cite{BSW14}. The structure of the paper is to present
information measures as we need them, rather than recalling all of them in one place.

Throughout we adopt the usual convention and define $f\left(  A\right)  $ for
a function $f$ and a positive semi-definite operator $A$ as follows:%
\begin{equation}
f\left(  A\right)  \equiv\sum_{i}f\left(  \lambda_{i}\right)  \left\vert
i\right\rangle \left\langle i\right\vert ,
\end{equation}
where $A=\sum_{i}\lambda_{i}\left\vert i\right\rangle \left\langle
i\right\vert $ is a spectral decomposition of $A$ such that $\lambda_{i}\neq0$
for all $i$. We let $\Pi_{A}$ denote the projection onto the support of $A$.

Another important technical tool for proving our main result is the operator
version of the Hadamard three-line theorem given in \cite{B13}, in particular,
the very slight modification stated in \cite{D14}. We note that the theorem
below is a variant of the Riesz-Thorin operator interpolation theorem (see,
e.g., \cite{BL76,RS75}).

\begin{theorem}
\label{thm:hadamard}Let%
\begin{equation}
S\equiv\left\{  z\in\mathbb{C}:0\leq\operatorname{Re}\left\{  z\right\}
\leq1\right\}  , \label{eq:strip}%
\end{equation}
and let $L\left(  \mathcal{H}\right)  $ be the space of bounded linear
operators acting on a Hilbert space $\mathcal{H}$. Let $G:S\rightarrow
L\left(  \mathcal{H}\right)  $ be a bounded map that is holomorphic on the
interior of $S$ and continuous on the boundary.\footnote{A map $G:S\rightarrow
L(\mathcal{H})$ is holomorphic (continuous, bounded) if the corresponding
functions to matrix entries are holomorphic (continuous, bounded).} Let
$\theta\in\left(  0,1\right)  $ and define $p_{\theta}$ by%
\begin{equation}
\frac{1}{p_{\theta}}=\frac{1-\theta}{p_{0}}+\frac{\theta}{p_{1}},
\label{eq:p-relation}%
\end{equation}
where $p_{0},p_{1}\in\lbrack1,\infty]$. For $k=0,1$ define%
\begin{equation}
M_{k}=\sup_{t\in\mathbb{R}}\left\Vert G\left(  k+it\right)  \right\Vert
_{p_{k}}.
\end{equation}
Then%
\begin{equation}
\left\Vert G\left(  \theta\right)  \right\Vert _{p_{\theta}}\leq
M_{0}^{1-\theta}M_{1}^{\theta}. \label{eq:hadamard-3-line}%
\end{equation}

\end{theorem}

\section{Bounds for a difference of quantum relative entropies}

\label{sec:main-result}This section presents our main result
(Theorem~\ref{thm:rel-ent}), which is a refinement of the monotonicity of
quantum relative entropy. For the lower bounds given in this paper, we take
states $\rho$ and $\sigma$ and the channel $\mathcal{N}$ to be as given in the
following definition:

\begin{definition}
\label{def:rho-sig-N}Let $\rho$ be a density operator and let $\sigma$ be a
positive semi-definite operator, each acting on a finite-dimensional Hilbert
space $\mathcal{H}_{S}$ and such that $\operatorname{supp}\left(  \rho\right)
\subseteq\operatorname{supp}\left(  \sigma\right)  $. Let $\mathcal{N}%
:L\left(  \mathcal{H}_{S}\right)  \rightarrow L\left(  \mathcal{H}_{B}\right)
$ be a quantum channel with finite-dimensional output Hilbert space
$\mathcal{H}_{B}$.
\end{definition}

A R\'{e}nyi generalization of a relative entropy difference is defined as
\cite{SBW14}%
\begin{equation}
\widetilde{\Delta}_{\alpha}\left(  \rho,\sigma,\mathcal{N}\right)  \equiv
\frac{2\alpha}{\alpha-1}\log\left\Vert \left(  \left[  \mathcal{N}\left(
\rho\right)  \right]  ^{\left(  1-\alpha\right)  /2\alpha}\left[
\mathcal{N}\left(  \sigma\right)  \right]  ^{\left(  \alpha-1\right)
/2\alpha}\otimes I_{E}\right)  U_{S\rightarrow BE}\sigma^{\left(
1-\alpha\right)  /2\alpha}\rho^{1/2}\right\Vert _{2\alpha},
\label{eq:renyi-diff}%
\end{equation}
where here and throughout this paper $\log$ denotes the natural logarithm,
$\alpha\in(0,1)\cup(1,\infty)$, and $U_{S\rightarrow BE}$ is an isometric
extension of the channel $\mathcal{N}$. That is, $U_{S\rightarrow BE}$ is a
linear isometry satisfying%
\begin{equation}
\operatorname{Tr}_{E}\left\{  U_{S\rightarrow BE}\left(  \cdot\right)
_{S}U_{S\rightarrow BE}^{\dag}\right\}  =\mathcal{N}\left(  \cdot\right)
,\ \ \ \ \ \,\,\,\,U_{S\rightarrow BE}^{\dag}U_{S\rightarrow BE}=I_{S}.
\end{equation}
All isometric extensions of a channel are related by an isometry acting on the
environment system~$E$, so that the definition in \eqref{eq:renyi-diff} is
invariant under any such choice. Recall also that the adjoint $\mathcal{N}%
^{\dag}$\ of a channel is given in terms of an isometric extension $U$ as%
\begin{equation}
\mathcal{N}^{\dag}\left(  \cdot\right)  =U^{\dag}\left(  \left(  \cdot\right)
\otimes I_{E}\right)  U.
\end{equation}
(This can be used to verify that the definition given in \eqref{eq:renyi-diff}
is the same as the definition given in \cite{SBW14}.)

The following limit is known for positive definite operators \cite[Section
6]{SBW14} and we provide a proof in Appendix~\ref{app:limit}\ that it holds
for $\rho$, $\sigma$, and $\mathcal{N}$ as given in
Definition~\ref{def:rho-sig-N}:%
\begin{equation}
\lim_{\alpha\rightarrow1}\widetilde{\Delta}_{\alpha}\left(  \rho
,\sigma,\mathcal{N}\right)  =D\left(  \rho\Vert\sigma\right)  -D\left(
\mathcal{N}\left(  \rho\right)  \Vert\mathcal{N}\left(  \sigma\right)
\right)  .\label{eq:rel-ent-diff-a-1}%
\end{equation}
It is one reason why we say that $\widetilde{\Delta}_{\alpha}\left(
\rho,\sigma,\mathcal{N}\right)  $ is a R\'{e}nyi generalization of a relative
entropy difference, in addition to the fact that $\widetilde{\Delta}_{\alpha
}\left(  \rho,\sigma,\mathcal{N}\right)  \geq0$ for all $\alpha\in
\lbrack1/2,1)\cup(1,\infty)$ \cite{DW15}. The quantum relative entropy
$D\left(  \omega\Vert\tau\right)  $ is defined for a density operator $\omega$
and a positive semi-definite operator $\tau$ as \cite{U62}%
\begin{equation}
D\left(  \omega\Vert\tau\right)  \equiv\operatorname{Tr}\left\{  \omega\left[
\log\omega-\log\tau\right]  \right\}  ,
\end{equation}
whenever $\operatorname{supp}\left(  \omega\right)  \subseteq
\operatorname{supp}\left(  \tau\right)  $, and by convention, it is defined to
be $+\infty$ otherwise. It is monotone with respect to quantum channels
\cite{Lindblad1975,U77}\ in the following sense:%
\begin{equation}
D\left(  \rho\Vert\sigma\right)  -D\left(  \mathcal{N}\left(  \rho\right)
\Vert\mathcal{N}\left(  \sigma\right)  \right)  \geq0.
\end{equation}
We refer to the quantity on the right-hand side of \eqref{eq:rel-ent-diff-a-1}
as a \textquotedblleft relative entropy difference.\textquotedblright

For $\alpha=1/2$, observe that%
\begin{align}
\widetilde{\Delta}_{1/2}\left(  \rho,\sigma,\mathcal{N}\right)   &
=-\log\left\Vert \left(  \left[  \mathcal{N}\left(  \rho\right)  \right]
^{1/2}\left[  \mathcal{N}\left(  \sigma\right)  \right]  ^{-1/2}\otimes
I_{E}\right)  U_{S\rightarrow BE}\sigma^{1/2}\rho^{1/2}\right\Vert _{1}^{2}\\
&  =-\log F\left(  \rho,\mathcal{R}_{\sigma,\mathcal{N}}^{P}\left(
\mathcal{N}\left(  \rho\right)  \right)  \right)  .
\end{align}
where $F\left(  \rho,\sigma\right)  \equiv\left\Vert \sqrt{\rho}\sqrt{\sigma
}\right\Vert _{1}^{2}$ is the quantum fidelity \cite{U73}\ and $\mathcal{R}%
_{\sigma,\mathcal{N}}^{P}$ is the Petz recovery map \cite{Petz1986,Petz1988}%
\ (see also \cite{BK02}) defined as%
\begin{equation}
\mathcal{R}_{\sigma,\mathcal{N}}^{P}\left(  \cdot\right)  \equiv\sigma
^{1/2}\mathcal{N}^{\dag}\left(  \left[  \mathcal{N}\left(  \sigma\right)
\right]  ^{-1/2}\left(  \cdot\right)  \left[  \mathcal{N}\left(
\sigma\right)  \right]  ^{-1/2}\right)  \sigma^{1/2}%
.\label{eq:Petz-channel-Rel-ent}%
\end{equation}
(See Appendix~\ref{app:Petz-channel}\ for a brief justification that
$\mathcal{R}_{\sigma,\mathcal{N}}^{P}$ is a completely positive
trace-non-increasing linear map and a quantum channel when acting on
$\operatorname{supp}\left(  \mathcal{N}\left(  \sigma\right)  \right)  $.)
From the definition, one can see that the fidelity possesses the following
properties:%
\begin{equation}
\sqrt{F}\left(  \omega_{XB},\tau_{XB}\right)  =\sum_{x}p_{X}\left(  x\right)
\sqrt{F}\left(  \omega_{x},\tau_{x}\right)  ,\label{eq:fid-flags}%
\end{equation}
where%
\begin{equation}
\omega_{XB}\equiv\sum_{x}p_{X}\left(  x\right)  \left\vert x\right\rangle
\left\langle x\right\vert _{X}\otimes\omega_{x},\ \ \ \ \tau_{XB}\equiv
\sum_{x}p_{X}\left(  x\right)  \left\vert x\right\rangle \left\langle
x\right\vert _{X}\otimes\tau_{x},\label{eq:flag-states}%
\end{equation}
$p_{X}$ is a probability distribution, $\left\{  \left\vert x\right\rangle
\right\}  $ is some orthonormal basis, and $\{\omega_{x}\}$ and $\{\tau_{x}\}$
are sets of states.

For the upper bounds given in this paper, the situation is a bit more
restrictive, and we take $\rho$, $\sigma$, and $\mathcal{N}$ as in the
following definition:

\begin{definition}
\label{def:rho-sig-N-2}Let $\rho_{SE^{\prime}}$ be a positive definite density
operator and let $\sigma_{SE^{\prime}}$ be a positive definite operator, each
acting on a finite-dimensional tensor-product Hilbert space $\mathcal{H}%
_{S}\otimes\mathcal{H}_{E^{\prime}}$. Let $\mathcal{N}$ be a quantum channel
given as follows:%
\begin{equation}
\mathcal{N}\left(  \theta_{SE^{\prime}}\right)  =\operatorname{Tr}_{E}\left\{
U_{SE^{\prime}\rightarrow BE}\theta_{SE^{\prime}}U_{SE^{\prime}\rightarrow
BE}^{\dag}\right\}  ,
\end{equation}
where $U_{SE^{\prime}\rightarrow BE}$ is a unitary operator taking
$\mathcal{H}_{S}\otimes\mathcal{H}_{E^{\prime}}$ to an isomorphic
finite-dimensional tensor-product Hilbert space $\mathcal{H}_{B}%
\otimes\mathcal{H}_{E}$, such that $\mathcal{N}\left(  \rho\right)  $ and
$\mathcal{N}\left(  \sigma\right)  $ are each positive definite and act
on~$\mathcal{H}_{B}$.
\end{definition}

Let $\rho$, $\sigma$, and $\mathcal{N}$ be as given in
Definition~\ref{def:rho-sig-N-2}. We require this restriction for the upper
bounds because in this case, we will be taking matrix inverses and need to
conclude statements such as the following one:%
\begin{equation}
\left[  \sigma^{-1/2}\mathcal{N}^{\dag}\left(  \left[  \mathcal{N}\left(
\sigma\right)  \right]  ^{1/2}\theta_{B}^{-1}\left[  \mathcal{N}\left(
\sigma\right)  \right]  ^{1/2}\right)  \sigma^{-1/2}\right]  ^{-1}%
=\sigma^{1/2}\mathcal{N}^{\dag}\left(  \left[  \mathcal{N}\left(
\sigma\right)  \right]  ^{-1/2}\theta_{B}\left[  \mathcal{N}\left(
\sigma\right)  \right]  ^{-1/2}\right)  \sigma^{1/2},
\end{equation}
where $\theta_{B}$ is positive definite. The equality above follows because in
this case%
\begin{equation}
\mathcal{N}^{\dag}\left(  \theta_{B}\right)  =U^{\dag}\left(  \theta
_{B}\otimes I_{E}\right)  U,
\end{equation}
with $U$ unitary, so that $\left[  \mathcal{N}^{\dag}\left(  \theta_{B}%
^{-1}\right)  \right]  ^{-1}=\mathcal{N}^{\dag}\left(  \theta_{B}\right)  $
(this equality need not hold if $\rho$, $\sigma$, and $\mathcal{N}$ are
allowed the more general form as in Definition~\ref{def:rho-sig-N}---i.e., a
matrix inverse does not commute with a partial trace). It then follows from
the method of proof given in \cite[Proposition 29]{BSW14} that the following
limit holds%
\begin{equation}
\lim_{\alpha\rightarrow\infty}\widetilde{\Delta}_{\alpha}\left(  \rho
,\sigma,\mathcal{N}\right)  =D_{\max}\left(  \rho\middle\Vert\mathcal{R}%
_{\sigma,\mathcal{N}}^{P}\left(  \mathcal{N}\left(  \rho\right)  \right)
\right)  ,\label{eq:D-max-rel-ent}%
\end{equation}
where%
\begin{equation}
D_{\max}\left(  \omega\Vert\tau\right)  \equiv\log\left\Vert \omega^{1/2}%
\tau^{-1}\omega^{1/2}\right\Vert _{\infty}=2\log\left\Vert \omega^{1/2}%
\tau^{-1/2}\right\Vert _{\infty}%
\end{equation}
is the max-relative entropy\ \cite{D09}. The quantity on the right-hand side
of \eqref{eq:D-max-rel-ent} was defined in \cite{DW15}, following directly
from the ideas presented in \cite{BSW14,SBW14}. From the definition, one can
see that the max-relative entropy possesses the following property:%
\begin{equation}
D_{\max}\left(  \omega_{XB}\Vert\tau_{XB}\right)  =\max_{x}D_{\max}\left(
\omega_{x}\Vert\tau_{x}\right)  ,\label{eq:max-flags}%
\end{equation}
where $\omega_{XB}$ and $\tau_{XB}$ are as in (\ref{eq:flag-states}).

We can now state the main theorem of this paper:

\begin{theorem}
\label{thm:rel-ent}Let $\rho$, $\sigma$, and $\mathcal{N}$ be as given in
Definition~\ref{def:rho-sig-N}. Then the following inequality holds%
\begin{equation}
-\log\left[  \sup_{t\in\mathbb{R}}F\left(  \rho,\mathcal{R}_{\sigma
,\mathcal{N}}^{P,t}\left(  \mathcal{N(}\rho)\right)  \right)  \right]  \leq
D\left(  \rho\Vert\sigma\right)  -D\left(  \mathcal{N(}\rho)\Vert
\mathcal{N(}\sigma)\right)  ,\label{eq:rel-ent-ineq}%
\end{equation}
where $\mathcal{R}_{\sigma,\mathcal{N}}^{P,t}$ is the following rotated Petz
recovery map:%
\begin{equation}
\mathcal{R}_{\sigma,\mathcal{N}}^{P,t}\left(  \cdot\right)  \equiv\left(
\mathcal{U}_{\sigma,t}\circ\mathcal{R}_{\sigma,\mathcal{N}}^{P}\circ
\mathcal{U}_{\mathcal{N}\left(  \sigma\right)  ,-t}\right)  \left(
\cdot\right)  ,\label{eq:rotated-Petz}%
\end{equation}
$\mathcal{R}_{\sigma,\mathcal{N}}^{P}$ is the Petz recovery map defined in
\eqref{eq:Petz-channel-Rel-ent}, and $\mathcal{U}_{\sigma,t}$ and
$\mathcal{U}_{\mathcal{N}\left(  \sigma\right)  ,-t}$ are partial isometric
maps defined from%
\begin{equation}
\mathcal{U}_{\omega,t}\left(  \cdot\right)  \equiv\omega^{it}\left(
\cdot\right)  \omega^{-it},\label{eq:unitaries}%
\end{equation}
with $\omega$ a positive semi-definite operator. If $\rho$, $\sigma$, and
$\mathcal{N}$ are as given in Definition~\ref{def:rho-sig-N-2}, then the
following inequality holds%
\begin{equation}
D\left(  \rho\Vert\sigma\right)  -D\left(  \mathcal{N(}\rho)\Vert
\mathcal{N(}\sigma)\right)  \leq\sup_{t\in\mathbb{R}}D_{\max}\left(
\rho\middle\Vert\mathcal{R}_{\sigma,\mathcal{N}}^{P,t}(\mathcal{N(}%
\rho))\right)  .\label{eq:UP-rel-ent}%
\end{equation}

\end{theorem}

\begin{proof}
We can prove this result by employing Theorem~\ref{thm:hadamard}. We first
establish the inequality in (\ref{eq:rel-ent-ineq}). Let $U_{S\rightarrow BE}$
be an isometric extension of the channel $\mathcal{N}$, which we abbreviate as
$U$ in what follows. Pick%
\begin{align}
G\left(  z\right)   &  \equiv\left(  \left[  \mathcal{N}\left(  \rho\right)
\right]  ^{z/2}\left[  \mathcal{N}\left(  \sigma\right)  \right]
^{-z/2}\otimes I_{E}\right)  U\sigma^{z/2}\rho^{1/2},\label{eq:G(z)}\\
p_{0} &  =2,\label{eq:p0}\\
p_{1} &  =1,\label{eq:p1}%
\end{align}
and $\theta\in\left(  0,1\right)  $, which fixes%
\begin{equation}
p_{\theta}=\frac{2}{1+\theta}.\label{eq:ptheta}%
\end{equation}
The operator valued-function $G\left(  z\right)  $ satisfies the conditions
needed to apply Theorem~\ref{thm:hadamard}.\footnote{Note that boundedness
follows from the finite-dimensional assumption---however the stronger bound
$\left\Vert G\left(  z\right)  \right\Vert _{\infty}\leq1$ holds for all $z\in
S$, where $S$ is defined in (\ref{eq:strip}) (this is a consequence of
(\ref{eq:leads-to-improved-bnd}), given that the quantum fidelity does not
exceed one).} For the choices in \eqref{eq:G(z)}-\eqref{eq:ptheta}, we find
that%
\begin{equation}
\left\Vert G\left(  \theta\right)  \right\Vert _{2/\left(  1+\theta\right)
}=\left\Vert \left(  \left[  \mathcal{N}\left(  \rho\right)  \right]
^{\theta/2}\left[  \mathcal{N}\left(  \sigma\right)  \right]  ^{-\theta
/2}\otimes I_{E}\right)  U\sigma^{\theta/2}\rho^{1/2}\right\Vert _{2/\left(
1+\theta\right)  },
\end{equation}%
\begin{align}
M_{0} &  =\sup_{t\in\mathbb{R}}\left\Vert G\left(  it\right)  \right\Vert
_{2}=\sup_{t\in\mathbb{R}}\left\Vert \left(  \left[  \mathcal{N}\left(
\rho\right)  \right]  ^{it/2}\left[  \mathcal{N}\left(  \sigma\right)
\right]  ^{-it/2}\otimes I_{E}\right)  U\sigma^{it}\rho^{1/2}\right\Vert
_{2}\nonumber\\
&  \leq\left\Vert \rho^{1/2}\right\Vert _{2}=1,\label{eq:M0}\\
M_{1} &  =\sup_{t\in\mathbb{R}}\left\Vert G\left(  1+it\right)  \right\Vert
_{1}\nonumber\\
&  =\sup_{t\in\mathbb{R}}\left\Vert \left(  \left[  \mathcal{N}\left(
\rho\right)  \right]  ^{\left(  1+it\right)  /2}\left[  \mathcal{N}\left(
\sigma\right)  \right]  ^{-\left(  1+it\right)  /2}\otimes I_{E}\right)
U\sigma^{\left(  1+it\right)  /2}\rho^{1/2}\right\Vert _{1}\nonumber\\
&  =\sup_{t\in\mathbb{R}}\left\Vert \left(  \left[  \mathcal{N}\left(
\rho\right)  \right]  ^{it/2}\left[  \mathcal{N}\left(  \rho\right)  \right]
^{1/2}\left[  \mathcal{N}\left(  \sigma\right)  \right]  ^{-it/2}\left[
\mathcal{N}\left(  \sigma\right)  \right]  ^{-1/2}\otimes I_{E}\right)
U\sigma^{1/2}\sigma^{it}\rho^{1/2}\right\Vert _{1}\nonumber\\
&  =\sup_{t\in\mathbb{R}}\left\Vert \left(  \left[  \mathcal{N}\left(
\rho\right)  \right]  ^{1/2}\left[  \mathcal{N}\left(  \sigma\right)  \right]
^{-it/2}\left[  \mathcal{N}\left(  \sigma\right)  \right]  ^{-1/2}\otimes
I_{E}\right)  U\sigma^{1/2}\sigma^{it}\rho^{1/2}\right\Vert _{1}\nonumber\\
&  =\sup_{t\in\mathbb{R}}\sqrt{F}\left(  \rho,\left(  \mathcal{U}%
_{\sigma,-t/2}\circ\mathcal{R}_{\sigma,\mathcal{N}}^{P}\circ\mathcal{U}%
_{\mathcal{N}\left(  \sigma\right)  ,t/2}\right)  \left(  \mathcal{N}\left(
\rho\right)  \right)  \right)  \nonumber\\
&  =\left[  \sup_{t\in\mathbb{R}}F\left(  \rho,\mathcal{R}_{\sigma
,\mathcal{N}}^{P,t}\left(  \mathcal{N}\left(  \rho\right)  \right)  \right)
\right]  ^{1/2}.\label{eq:M1}%
\end{align}
Then we can apply \eqref{eq:hadamard-3-line} to conclude that%
\begin{equation}
\left\Vert \left(  \left[  \mathcal{N(}\rho)\right]  ^{\theta/2}\left[
\mathcal{N(}\sigma)\right]  ^{-\theta/2}\otimes I_{E}\right)  U\sigma
^{\theta/2}\rho^{1/2}\right\Vert _{2/\left(  1+\theta\right)  }\leq\left[
\sup_{t\in\mathbb{R}}F\left(  \rho,\mathcal{R}_{\sigma,\mathcal{N}}%
^{P,t}\left(  \mathcal{N(}\rho)\right)  \right)  \right]  ^{\theta
/2}.\label{eq:leads-to-improved-bnd}%
\end{equation}
Taking a negative logarithm gives%
\begin{equation}
-\log\left[  \sup_{t\in\mathbb{R}}F\left(  \rho,\mathcal{R}_{\sigma
,\mathcal{N}}^{P,t}\left(  \mathcal{N(}\rho)\right)  \right)  \right]
\leq-\frac{2}{\theta}\log\left\Vert \left(  \left[  \mathcal{N(}\rho)\right]
^{\theta/2}\left[  \mathcal{N(}\sigma)\right]  ^{-\theta/2}\otimes
I_{E}\right)  U\sigma^{\theta/2}\rho^{1/2}\right\Vert _{2/\left(
1+\theta\right)  }.\label{eq:critical-inequality}%
\end{equation}
Letting $\theta=\left(  1-\alpha\right)  /\alpha$, we see that this is the
same as%
\begin{equation}
-\log\left[  \sup_{t\in\mathbb{R}}F\left(  \rho,\mathcal{R}_{\sigma
,\mathcal{N}}^{P,t}\left(  \mathcal{N}\left(  \rho\right)  \right)  \right)
\right]  \leq\widetilde{\Delta}_{\alpha}\left(  \rho,\sigma,\mathcal{N}%
\right)  .\label{eq:alpha-bound}%
\end{equation}
Since the inequality in \eqref{eq:critical-inequality} holds for all
$\theta\in\left(  0,1\right)  $ and thus \eqref{eq:alpha-bound} holds for all
$\alpha\in\left(  1/2,1\right)  $, we can take the limit as $\alpha\nearrow1$
and apply \eqref{eq:rel-ent-diff-a-1} to conclude that \eqref{eq:rel-ent-ineq} holds.

We now establish the inequality in (\ref{eq:UP-rel-ent}) for $\rho$, $\sigma$,
and $\mathcal{N}$ as given in Definition~\ref{def:rho-sig-N-2}. Note that in
this case, $U$ is a unitary. Pick%
\begin{align}
G\left(  z\right)   &  \equiv\left(  \left[  \mathcal{N}\left(  \rho\right)
\right]  ^{-z/2}\left[  \mathcal{N}\left(  \sigma\right)  \right]
^{z/2}\otimes I_{E}\right)  U\sigma^{-z/2}\rho^{1/2},\\
p_{0}  &  =2,\\
p_{1}  &  =\infty,
\end{align}
and $\theta\in\left(  0,1\right)  $, which fixes%
\begin{equation}
p_{\theta}=\frac{2}{1-\theta}.
\end{equation}
The operator valued-function $G\left(  z\right)  $ satisfies the conditions
needed to apply Theorem~\ref{thm:hadamard}. We then find that $M_{0}=1$ as
before, and%
\begin{equation}
\left\Vert G\left(  \theta\right)  \right\Vert _{2/\left(  1-\theta\right)
}=\left\Vert \left(  \left[  \mathcal{N}\left(  \rho\right)  \right]
^{-\theta/2}\left[  \mathcal{N}\left(  \sigma\right)  \right]  ^{\theta
/2}\otimes I_{E}\right)  U\sigma^{-\theta/2}\rho^{1/2}\right\Vert _{2/\left(
1-\theta\right)  },
\end{equation}%
\begin{align}
M_{1}  &  =\sup_{t\in\mathbb{R}}\left\Vert G\left(  1+it\right)  \right\Vert
_{\infty}\nonumber\\
&  =\sup_{t\in\mathbb{R}}\left\Vert \left(  \left[  \mathcal{N}\left(
\rho\right)  \right]  ^{-\left(  1+it\right)  /2}\left[  \mathcal{N}\left(
\sigma\right)  \right]  ^{\left(  1+it\right)  /2}\otimes I_{E}\right)
U\sigma^{-\left(  1+it\right)  /2}\rho^{1/2}\right\Vert _{\infty}\nonumber\\
&  =\sup_{t\in\mathbb{R}}\left\Vert \left(  \left[  \mathcal{N}\left(
\rho\right)  \right]  ^{-it/2}\left[  \mathcal{N}\left(  \rho\right)  \right]
^{-1/2}\left[  \mathcal{N}\left(  \sigma\right)  \right]  ^{it/2}\left[
\mathcal{N}\left(  \sigma\right)  \right]  ^{1/2}\otimes I_{E}\right)
U\sigma^{-1/2}\sigma^{-it/2}\rho^{1/2}\right\Vert _{\infty}\nonumber\\
&  =\sup_{t\in\mathbb{R}}\left\Vert \left(  \left[  \mathcal{N}\left(
\rho\right)  \right]  ^{-1/2}\left[  \mathcal{N}\left(  \sigma\right)
\right]  ^{it/2}\left[  \mathcal{N}\left(  \sigma\right)  \right]
^{1/2}\otimes I_{E}\right)  U\sigma^{-1/2}\sigma^{-it/2}\rho^{1/2}\right\Vert
_{\infty}\nonumber\\
&  =\left[  \exp\sup_{t\in\mathbb{R}}D_{\max}\left(  \rho\middle\Vert\left(
\mathcal{U}_{\sigma,-t}\circ\mathcal{R}_{\sigma,\mathcal{N}}^{P}%
\circ\mathcal{U}_{\mathcal{N}\left(  \sigma\right)  ,t}\right)  \left(
\mathcal{N}\left(  \rho\right)  \right)  \right)  \right]  ^{1/2}\nonumber\\
&  =\left[  \exp\sup_{t\in\mathbb{R}}D_{\max}\left(  \rho\middle\Vert
\mathcal{R}_{\sigma,\mathcal{N}}^{P,t}\left(  \mathcal{N}\left(  \rho\right)
\right)  \right)  \right]  ^{1/2}.
\end{align}
Then we can apply \eqref{eq:hadamard-3-line} to conclude that%
\begin{equation}
\left\Vert \left(  \left[  \mathcal{N(}\rho)\right]  ^{-\theta/2}\left[
\mathcal{N(}\sigma)\right]  ^{\theta/2}\otimes I_{E}\right)  U\sigma
^{-\theta/2}\rho^{1/2}\right\Vert _{2/\left(  1-\theta\right)  }\leq\left[
\exp\sup_{t\in\mathbb{R}}D_{\max}\left(  \rho\middle\Vert\mathcal{R}%
_{\sigma,\mathcal{N}}^{P,t}\left(  \mathcal{N}\left(  \rho\right)  \right)
\right)  \right]  ^{\theta/2}.
\end{equation}
Taking a logarithm gives%
\begin{equation}
\frac{2}{\theta}\log\left\Vert \left(  \left[  \mathcal{N(}\rho)\right]
^{-\theta/2}\left[  \mathcal{N(}\sigma)\right]  ^{\theta/2}\otimes
I_{E}\right)  U\sigma^{-\theta/2}\rho^{1/2}\right\Vert _{2/\left(
1-\theta\right)  }\leq\sup_{t\in\mathbb{R}}D_{\max}\left(  \rho\middle\Vert
\mathcal{R}_{\sigma,\mathcal{N}}^{P,t}\left(  \mathcal{N}\left(  \rho\right)
\right)  \right)  . \label{eq:critical-inequality-2}%
\end{equation}
Letting $\theta=\left(  \alpha-1\right)  /\alpha$, we see that this is the
same as%
\begin{equation}
\widetilde{\Delta}_{\alpha}\left(  \rho,\sigma,\mathcal{N}\right)  \leq
\sup_{t\in\mathbb{R}}D_{\max}\left(  \rho\middle\Vert\mathcal{R}%
_{\sigma,\mathcal{N}}^{P,t}\left(  \mathcal{N}\left(  \rho\right)  \right)
\right)  . \label{eq:alpha-bound-2}%
\end{equation}
Since the inequality in \eqref{eq:critical-inequality-2} holds for all
$\theta\in\left(  0,1\right)  $ and thus \eqref{eq:alpha-bound-2} holds for
all $\alpha\in\left(  1,\infty\right)  $, we can take the limit as
$\alpha\searrow1$ and apply \eqref{eq:rel-ent-diff-a-1} to conclude that
\eqref{eq:UP-rel-ent} holds.
\end{proof}

\begin{remark}
\label{rem:perfect-recovery}We cannot necessarily conclude which value of $t$
is optimal in Theorem~\ref{thm:rel-ent}. However, it is clear that the partial
isometric map $\mathcal{U}_{\omega,t}$ preserves the density operator $\omega$
(i.e., that the partial isometry $\omega^{it}$ is diagonal in the eigenbasis
of $\omega$). Furthermore, the optimal value of $t$ could have a dependence on
the state $\rho$, which is undesirable for some applications such as
approximate quantum error correction.
\end{remark}

\begin{remark}
\label{rem:perfect-rec-sig}Any recovery map of the form $\mathcal{R}%
_{\sigma,\mathcal{N}}^{P,t}$ perfectly recovers $\sigma$ from $\mathcal{N}%
\left(  \sigma\right)  $, in the sense that%
\begin{equation}
\left(  \mathcal{U}_{\sigma,t}\circ\mathcal{R}_{\sigma,\mathcal{N}}^{P}%
\circ\mathcal{U}_{\mathcal{N}\left(  \sigma\right)  ,-t}\right)  \left(
\mathcal{N}\left(  \sigma\right)  \right)  =\sigma,
\end{equation}
because%
\begin{equation}
\mathcal{U}_{\mathcal{N}\left(  \sigma\right)  ,-t}\left(  \mathcal{N}\left(
\sigma\right)  \right)  =\mathcal{N}\left(  \sigma\right)
,\ \ \ \ \mathcal{R}_{\sigma,\mathcal{N}}^{P}\left(  \mathcal{N}\left(
\sigma\right)  \right)  =\sigma,\ \ \ \ \mathcal{U}_{\sigma,t}\left(
\sigma\right)  =\sigma.
\end{equation}
This answers an open question discussed in \cite{BLW14}. In particular, we can
say that there is a map $\mathcal{R}_{\sigma,\mathcal{N}}^{P,t}$ that
perfectly recovers $\sigma$ from $\mathcal{N}\left(  \sigma\right)  $, while
having a performance limited by \eqref{eq:rel-ent-ineq} when recovering $\rho$
from $\mathcal{N}\left(  \rho\right)  $.
\end{remark}

\begin{remark}
From \eqref{eq:alpha-bound}\ in\ the proof given above, we can conclude that%
\begin{equation}
\widetilde{\Delta}_{\alpha}\left(  \rho,\sigma,\mathcal{N}\right)  \geq
-\log\left[  \sup_{t\in\mathbb{R}}F\left(  \rho,\mathcal{R}_{\sigma
,\mathcal{N}}^{P,t}\left(  \mathcal{N}\left(  \rho\right)  \right)  \right)
\right]  ,
\end{equation}
for all $\alpha\in\left(  1/2,1\right)  $. This inequality improves upon a
previous result from \cite{DW15}, which established that $\widetilde{\Delta
}_{\alpha}$ is non-negative for the same range of $\alpha$.
One also sees that $\widetilde{\Delta}_{\alpha}\left(  \rho,\sigma,\mathcal{N}\right) = 0$ implies that the channel $\mathcal{N}$ is sufficient for $\rho$ and
$\sigma$, in the sense that this condition implies the existence of a recovery map which perfectly recovers $\rho$ from $\mathcal{N}(\rho)$ and $\sigma$ from $\mathcal{N}(\sigma)$ .
\end{remark}

\section{Functoriality}

\label{sec:func}For a fixed $t$, recovery maps $\mathcal{R}_{\sigma
,\mathcal{N}}^{P,t}$ of the form in \eqref{eq:rotated-Petz} satisfy several
desirable \textquotedblleft functoriality\textquotedblright\ properties stated
in \cite{LW14}, in addition to the property stated in
Remark~\ref{rem:perfect-rec-sig}. These include normalization, parallel
composition, and serial composition, which we discuss in the following subsections.

\subsection{Normalization}

If there is in fact no noise, so that $\mathcal{N}=\operatorname{id}$, then we
would expect the recovery map to be equal to the identity channel as well.
This property is known as normalization \cite{LW14}, and we confirm it below
for all maps $\mathcal{R}_{\sigma,\mathcal{N}}^{P,t}$ of the form in
\eqref{eq:rotated-Petz} when $\mathcal{N}=\operatorname{id}$:%
\begin{align}
\mathcal{R}_{\sigma,\operatorname{id}}^{P,t}\left(  \cdot\right)   &  =\left(
\mathcal{U}_{\sigma,t}\circ\mathcal{R}_{\sigma,\operatorname{id}}^{P}%
\circ\mathcal{U}_{\operatorname{id}\left(  \sigma\right)  ,-t}\right)  \left(
\cdot\right)  \nonumber\\
&  =\sigma^{it}\sigma^{1/2}\operatorname{id}^{\dag}\left(  \left[
\operatorname{id}\left(  \sigma\right)  \right]  ^{-1/2}\left[
\operatorname{id}\left(  \sigma\right)  \right]  ^{-it}\left(  \cdot\right)
\left[  \operatorname{id}\left(  \sigma\right)  \right]  ^{it}\left[
\operatorname{id}\left(  \sigma\right)  \right]  ^{-1/2}\right)  \sigma
^{1/2}\sigma^{-it}\nonumber\\
&  =\sigma^{it}\sigma^{1/2}\left(  \sigma^{-1/2}\sigma^{-it}\left(
\cdot\right)  \sigma^{it}\sigma^{-1/2}\right)  \sigma^{1/2}\sigma
^{-it}\nonumber\\
&  =\Pi_{\sigma}\left(  \cdot\right)  \Pi_{\sigma}.
\end{align}
Thus, when $\sigma$ is positive definite, the recovery map is the identity channel.

\subsection{Parallel Composition}

If the $\sigma$ operator is a tensor product $\sigma_{1}\otimes\sigma_{2}$ and
the channel $\mathcal{N}$ is as well $\mathcal{N}_{1}\otimes\mathcal{N}_{2}$
(respecting the same tensor-product structure), then it would be desirable for
the recovery map to be a tensor product respecting this structure. This
property is known as parallel composition \cite{LW14}, and we confirm it below
for all maps $\mathcal{R}_{\sigma,\mathcal{N}}^{P,t}$ of the form in
\eqref{eq:rotated-Petz} when $\sigma=\sigma_{1}\otimes\sigma_{2}$ and
$\mathcal{N}=\mathcal{N}_{1}\otimes\mathcal{N}_{2}$. In fact, this property is
a consequence of the following:%
\begin{equation}
\mathcal{U}_{\sigma_{1}\otimes\sigma_{2},t}\left(  \cdot\right)  =\left[
\sigma_{1}\otimes\sigma_{2}\right]  ^{it}\left(  \cdot\right)  \left[
\sigma_{1}\otimes\sigma_{2}\right]  ^{-it}=\left[  \sigma_{1}^{it}%
\otimes\sigma_{2}^{it}\right]  \left(  \cdot\right)  \left[  \sigma_{1}%
^{-it}\otimes\sigma_{2}^{-it}\right]  =\left(  \mathcal{U}_{\sigma_{1}%
,t}\otimes\mathcal{U}_{\sigma_{2},t}\right)  \left(  \cdot\right)
,\label{eq:unitary-parallel}%
\end{equation}%
\begin{equation}
\mathcal{R}_{\sigma_{1}\otimes\sigma_{2},\mathcal{N}_{1}\otimes\mathcal{N}%
_{2}}^{P}\left(  \cdot\right)  =\left(  \mathcal{R}_{\sigma_{1},\mathcal{N}%
_{1}}^{P}\otimes\mathcal{R}_{\sigma_{2},\mathcal{N}_{2}}^{P}\right)  \left(
\cdot\right)  ,\label{eq:Petz-parallel}%
\end{equation}
where \eqref{eq:Petz-parallel} follows because%
\begin{align}
\left[  \sigma_{1}\otimes\sigma_{2}\right]  ^{1/2} &  =\sigma_{1}^{1/2}%
\otimes\sigma_{2}^{1/2},\ \ \ \ \ \left(  \mathcal{N}_{1}\otimes
\mathcal{N}_{2}\right)  ^{\dag}=\mathcal{N}_{1}^{\dag}\otimes\mathcal{N}%
_{2}^{\dag},\nonumber\\
\left(  \left(  \mathcal{N}_{1}\otimes\mathcal{N}_{2}\right)  \left(
\sigma_{1}\otimes\sigma_{2}\right)  \right)  ^{-1/2} &  =\left[
\mathcal{N}_{1}\left(  \sigma_{1}\right)  \right]  ^{-1/2}\otimes\left[
\mathcal{N}_{2}\left(  \sigma_{2}\right)  \right]  ^{-1/2}.
\end{align}
The following equality results from similar reasoning as in
\eqref{eq:unitary-parallel}:%
\begin{equation}
\mathcal{U}_{\left(  \mathcal{N}_{1}\otimes\mathcal{N}_{2}\right)  \left(
\sigma_{1}\otimes\sigma_{2}\right)  ,t}\left(  \cdot\right)  =\left(
\mathcal{U}_{\mathcal{N}_{1}\left(  \sigma_{1}\right)  ,t}\otimes
\mathcal{U}_{\mathcal{N}_{2}\left(  \sigma_{2}\right)  ,t}\right)  \left(
\cdot\right)  .
\end{equation}
Putting everything together, we find that parallel composition holds%
\begin{equation}
\mathcal{R}_{\sigma_{1}\otimes\sigma_{2},\mathcal{N}_{1}\otimes\mathcal{N}%
_{2}}^{P,t}\left(  \cdot\right)  =\left(  \mathcal{R}_{\sigma_{1}%
,\mathcal{N}_{1}}^{P,t}\otimes\mathcal{R}_{\sigma_{2},\mathcal{N}_{2}}%
^{P,t}\right)  \left(  \cdot\right)  .
\end{equation}

\subsection{Serial Composition}

If the channel $\mathcal{N}$ consists of the serial composition of two
channels $\mathcal{N}_{1}$ and $\mathcal{N}_{2}$, so that $\mathcal{N=N}%
_{2}\circ\mathcal{N}_{1}$, then it would be desirable for the recovery map to
consist of recovering from the last channel first and then from the first
channel. This property is known as serial composition \cite{LW14}, and we
confirm it below for all maps $\mathcal{R}_{\sigma,\mathcal{N}}^{P,t}$ of the
form in \eqref{eq:rotated-Petz} when $\mathcal{N}=\mathcal{N}_{2}%
\circ\mathcal{N}_{1}$. It is a consequence of the fact that%
\begin{equation}
\left(  \mathcal{N}_{2}\circ\mathcal{N}_{1}\right)  ^{\dag}=\mathcal{N}%
_{1}^{\dag}\circ\mathcal{N}_{2}^{\dag},
\end{equation}
so that%
\begin{align}
&  \mathcal{R}_{\sigma,\mathcal{N}_{2}\circ\mathcal{N}_{1}}^{P}\left(
\cdot\right)  \nonumber\\
&  =\sigma^{1/2}\left(  \mathcal{N}_{1}^{\dag}\circ\mathcal{N}_{2}^{\dag
}\right)  \left(  \left[  \left(  \mathcal{N}_{2}\circ\mathcal{N}_{1}\right)
\left(  \sigma\right)  \right]  ^{-1/2}\left(  \cdot\right)  \left[  \left(
\mathcal{N}_{2}\circ\mathcal{N}_{1}\right)  \left(  \sigma\right)  \right]
^{-1/2}\right)  \sigma^{1/2}\nonumber\\
&  =\sigma^{\frac{1}{2}}\mathcal{N}_{1}^{\dag}\left(  \mathcal{N}_{1}\left(
\sigma\right)  ^{-\frac{1}{2}}\left(  \mathcal{N}_{1}\left(  \sigma\right)
^{\frac{1}{2}}\left[  \mathcal{N}_{2}^{\dag}\left(  \left[  \mathcal{N}%
_{2}\left(  \mathcal{N}_{1}\left(  \sigma\right)  \right)  \right]
^{-\frac{1}{2}}\left(  \cdot\right)  \left[  \mathcal{N}_{2}\left(
\mathcal{N}_{1}\left(  \sigma\right)  \right)  \right]  ^{-\frac{1}{2}%
}\right)  \right]  \mathcal{N}_{1}\left(  \sigma\right)  ^{\frac{1}{2}%
}\right)  \mathcal{N}_{1}\left(  \sigma\right)  ^{-\frac{1}{2}}\right)
\sigma^{\frac{1}{2}}\nonumber\\
&  =\left(  \mathcal{R}_{\sigma,\mathcal{N}_{1}}^{P}\circ\mathcal{R}%
_{\mathcal{N}_{1}\left(  \sigma\right)  ,\mathcal{N}_{2}}^{P}\right)  \left(
\cdot\right)  .
\end{align}
Then the serial composition property follows because%
\begin{align}
\mathcal{R}_{\sigma,\mathcal{N}_{2}\circ\mathcal{N}_{1}}^{P,t}\left(
\cdot\right)   &  =\left(  \mathcal{U}_{\sigma,t}\circ\mathcal{R}%
_{\sigma,\mathcal{N}_{2}\circ\mathcal{N}_{1}}^{P}\circ\mathcal{U}_{\left(
\mathcal{N}_{2}\circ\mathcal{N}_{1}\right)  \left(  \sigma\right)
,-t}\right)  \left(  \cdot\right)  \nonumber\\
&  =\left(  \mathcal{U}_{\sigma,t}\circ\mathcal{R}_{\sigma,\mathcal{N}_{1}%
}^{P}\circ\mathcal{R}_{\mathcal{N}_{1}\left(  \sigma\right)  ,\mathcal{N}_{2}%
}^{P}\circ\mathcal{U}_{\left(  \mathcal{N}_{2}\circ\mathcal{N}_{1}\right)
\left(  \sigma\right)  ,-t}\right)  \left(  \cdot\right)  \nonumber\\
&  =\left(  \mathcal{U}_{\sigma,t}\circ\mathcal{R}_{\sigma,\mathcal{N}_{1}%
}^{P}\circ\mathcal{U}_{\mathcal{N}_{1}\left(  \sigma\right)  ,-t}%
\circ\mathcal{U}_{\mathcal{N}_{1}\left(  \sigma\right)  ,t}\mathcal{R}%
_{\mathcal{N}_{1}\left(  \sigma\right)  ,\mathcal{N}_{2}}^{P}\circ
\mathcal{U}_{\left(  \mathcal{N}_{2}\circ\mathcal{N}_{1}\right)  \left(
\sigma\right)  ,-t}\right)  \left(  \cdot\right)  \nonumber\\
&  =\left(  \mathcal{R}_{\sigma,\mathcal{N}_{1}}^{P,t}\circ\mathcal{R}%
_{\mathcal{N}_{1}\left(  \sigma\right)  ,\mathcal{N}_{2}}^{P,t}\right)
\left(  \cdot\right)  ,
\end{align}
where $\mathcal{R}_{\mathcal{N}_{1}\left(  \sigma\right)  ,\mathcal{N}_{2}%
}^{P,t}$ is the map that recovers from $\mathcal{N}_{2}$ and $\mathcal{R}%
_{\sigma,\mathcal{N}_{1}}^{P,t}$ is the map that recovers from~$\mathcal{N}%
_{1}$.

\section{Consequences and applications of Theorem~\ref{thm:rel-ent}}

\label{sec:corollaries}Theorem~\ref{thm:rel-ent} leads to a strengthening of
many entropy inequalities, including strong subadditivity of quantum entropy,
concavity of conditional entropy, joint convexity of relative entropy,
non-negativity of quantum discord, the Holevo bound, and multipartite
information inequalities. We list these as corollaries and give brief proofs
for them in the following subsections. Furthermore, there are potential
applications to approximate quantum error correction as well, which we
discuss. Some of the observations given below have been made before in
previous papers as either conjectures or concrete results
\cite{BSW14,SBW14,SW14,FR14,SOR15} (with many of them becoming concrete after
the posting of \cite{FR14}), but in many cases, Theorem~\ref{thm:rel-ent}%
\ allows us to make more precise statements due to the structure of the
recovery map $\mathcal{R}_{\sigma,\mathcal{N}}^{P,t}$ in
\eqref{eq:rotated-Petz} and its functoriality properties discussed in the
previous section.

\subsection{Strong Subadditivity}

The conditional quantum mutual information of a tripartite state $\rho_{ABC}$
is defined as%
\begin{equation}
I\left(  A;B|C\right)  _{\rho}\equiv H\left(  AC\right)  _{\rho}+H\left(
BC\right)  _{\rho}-H\left(  C\right)  _{\rho}-H\left(  ABC\right)  _{\rho},
\end{equation}
where $H(F)_{\sigma}\equiv-\operatorname{Tr}\{\sigma_{F}\log\sigma_{F}\}$ is
the von Neumann entropy of a density operator $\sigma_{F}$. Strong
subadditivity is the statement that $I\left(  A;B|C\right)  _{\rho}\geq0$ for
all tripartite states $\rho_{ABC}$ \cite{LR73,PhysRevLett.30.434}.

Corollary~\ref{thm:CMI} below gives an improvement of strong subadditivity, in
addition to providing an upper bound on conditional mutual information. It is
a direct consequence of Theorem~\ref{thm:rel-ent}\ after choosing%
\begin{equation}
\rho=\rho_{ABC},\ \ \ \ \sigma=\rho_{AC}\otimes I_{B},\ \ \ \ \mathcal{N}%
=\operatorname{Tr}_{A},\label{eq:cmi-rho}%
\end{equation}
so that%
\begin{align}
\mathcal{N}\left(  \rho\right)   &  =\rho_{BC},\ \ \ \ \ \ \mathcal{N}\left(
\sigma\right)  =\rho_{C}\otimes I_{B},\\
D\left(  \rho\Vert\sigma\right)  -D\left(  \mathcal{N}\left(  \rho\right)
\Vert\mathcal{N}\left(  \sigma\right)  \right)   &  =D\left(  \rho_{ABC}%
\Vert\rho_{AC}\otimes I_{B}\right)  -D\left(  \rho_{BC}\Vert\rho_{C}\otimes
I_{B}\right)  \nonumber\\
&  =I\left(  A;B|C\right)  _{\rho},\\
\mathcal{N}^{\dag}\left(  \cdot\right)   &  =\left(  \cdot\right)  \otimes
I_{A},\\
\mathcal{R}_{\sigma,\mathcal{N}}^{P}\left(  \cdot\right)   &  =\sigma
^{1/2}\mathcal{N}^{\dag}\left(  \left[  \mathcal{N}\left(  \sigma\right)
\right]  ^{-1/2}\left(  \cdot\right)  \left[  \mathcal{N}\left(
\sigma\right)  \right]  ^{-1/2}\right)  \sigma^{1/2}\nonumber\\
&  =\rho_{AC}^{1/2}\left[  \rho_{C}^{-1/2}\left(  \cdot\right)  \rho
_{C}^{-1/2}\otimes I_{A}\right]  \rho_{AC}^{1/2}\nonumber\\
&  \equiv\mathcal{R}_{C\rightarrow AC}^{P}\left(  \cdot\right)
,\label{eq:Petz-channel-CMI}%
\end{align}
where $\mathcal{R}_{C\rightarrow AC}^{P}$ is a special case of the Petz
recovery map \cite{Petz1986,Petz1988} (see also \cite{LW14}).
%Note that $\rho
%$, $\sigma$, and $\mathcal{N}$ as chosen above have the form given in
%Definition~\ref{def:rho-sig-N-2}.

\begin{corollary}
\label{thm:CMI}Let $\rho_{ABC}$ be a density operator acting on a
finite-dimensional Hilbert space $\mathcal{H}_{A}\otimes\mathcal{H}_{B}%
\otimes\mathcal{H}_{C}$. Then the following inequality holds%
\begin{equation}
-\log\left[  \sup_{t\in\mathbb{R}}F\left(  \rho_{ABC},\mathcal{R}%
_{C\rightarrow AC}^{P,t}\left(  \rho_{BC}\right)  \right)  \right]  \leq
I\left(  A;B|C\right)  _{\rho},\label{eq:main-result}%
\end{equation}
where $\mathcal{R}_{C\rightarrow AC}^{P,t}$ is the following rotated Petz
recovery map:%
\begin{equation}
\mathcal{R}_{C\rightarrow AC}^{P,t}\left(  \cdot\right)  \equiv\left(
\mathcal{U}_{\rho_{AC},t}\circ\mathcal{R}_{C\rightarrow AC}^{P}\circ
\mathcal{U}_{\rho_{C},-t}\right)  \left(  \cdot\right)
,\label{eq:CMI-petz-recovery}%
\end{equation}
the Petz recovery map $\mathcal{R}_{C\rightarrow AC}^{P}$ is defined in
\eqref{eq:Petz-channel-CMI},\ and the partial isometric maps $\mathcal{U}%
_{\rho_{AC},t}$ and $\mathcal{U}_{\rho_{C},-t}$ are defined from
\eqref{eq:unitaries}. If $\rho_{ABC}$ is positive definite, then the following
inequality holds as well:%
\begin{equation}
I\left(  A;B|C\right)  _{\rho}\leq\sup_{t\in\mathbb{R}}D_{\max}\left(
\rho_{ABC}\middle\Vert\mathcal{R}_{C\rightarrow AC}^{P,t}\left(  \rho
_{BC}\right)  \right)  .
\end{equation}

\end{corollary}

\begin{remark}
A lower bound on $I\left(  A;B|C\right)  _{\rho}$ similar to that in
Corollary~\ref{thm:CMI} has already been identified \cite{FRSS15}\ by putting
together various statements from \cite{FR14,SOR15}. One needs to examine the
discussion surrounding Eqs.~(120)-(122) in \cite{SOR15} and make several
further arguments\ in order to arrive at this conclusion. See Remark~2.5 of
\cite{SOR15}.
\end{remark}

\begin{remark}
We note that $\widetilde{\Delta}_{\alpha}\left(  \rho,\sigma,\mathcal{N}%
\right)  $ for the choices in \eqref{eq:cmi-rho} reduces to the R\'{e}nyi
conditional mutual information from \cite{BSW14}:%
\begin{equation}
\widetilde{I}_{\alpha}\left(  A;B|C\right)  _{\rho}\equiv\frac{2\alpha}%
{\alpha-1}\log\left\Vert \rho_{BC}^{\left(  1-\alpha\right)  /2\alpha}\rho
_{C}^{\left(  \alpha-1\right)  /2\alpha}\rho_{AC}^{\left(  1-\alpha\right)
/2\alpha}\rho_{ABC}^{1/2}\right\Vert _{2\alpha},\label{eq:Renyi-CMI}%
\end{equation}
as observed in \cite{SBW14}. Thus, from the inequality in
\eqref{eq:alpha-bound}, we can conclude that%
\begin{equation}
\widetilde{I}_{\alpha}\left(  A;B|C\right)  _{\rho}\geq-\log\left[  \sup
_{t\in\mathbb{R}}F\left(  \rho_{ABC},\mathcal{R}_{C\rightarrow AC}%
^{P,t}\left(  \rho_{BC}\right)  \right)  \right]  ,
\end{equation}
for all $\alpha\in\left(  1/2,1\right)  $. This inequality improves upon a
previous result from \cite{BSW14}, which established that $\widetilde
{I}_{\alpha}$ is non-negative for the same range of $\alpha$.
\end{remark}

\begin{remark}
A statement similar to the first two sentences of
Remark~\ref{rem:perfect-recovery} applies as well to the partial isometric
maps in Corollary~\ref{thm:CMI}. Furthermore, the parameter $t$ has a
dependence on the global state $\rho_{ABC}$, so that the recovery map given
above does not possess the universality property discussed in \cite{SOR15}.
Regardless, the structure of the unitaries is sufficient for us to conclude
that any recovery map of the form $\mathcal{U}_{\rho_{AC},t}\circ
\mathcal{R}_{C\rightarrow AC}^{P}\circ\mathcal{U}_{\rho_{C},-t}$ always
perfectly recovers the state $\rho_{AC}$ from~$\rho_{C}$:%
\begin{equation}
\left(  \mathcal{U}_{\rho_{AC},t}\circ\mathcal{R}_{C\rightarrow AC}^{P}%
\circ\mathcal{U}_{\rho_{C},-t}\right)  \left(  \rho_{C}\right)  =\rho_{AC},
\end{equation}
because%
\begin{equation}
\mathcal{U}_{\rho_{C},-t}\left(  \rho_{C}\right)  =\rho_{C}%
,\ \ \ \ \mathcal{R}_{C\rightarrow AC}^{P}\left(  \rho_{C}\right)  =\rho
_{AC},\ \ \ \ \mathcal{U}_{\rho_{AC},t}\left(  \rho_{AC}\right)  =\rho_{AC}.
\end{equation}

\end{remark}

\subsection{Concavity of conditional quantum entropy}

The ideas in this section follow a line of thought developed in \cite{BLW14}.
Let $\mathcal{E}\equiv\left\{  p_{X}\left(  x\right)  ,\rho_{AB}^{x}\right\}
$ be an ensemble of bipartite quantum states with expectation%
\begin{equation}
\overline{\rho}_{AB}\equiv\sum_{x}p_{X}\left(  x\right)  \rho_{AB}^{x}.
\end{equation}
Concavity of conditional entropy is the statement that%
\begin{equation}
H\left(  A|B\right)  _{\overline{\rho}}\geq\sum_{x}p_{X}\left(  x\right)
H\left(  A|B\right)  _{\rho^{x}},
\end{equation}
where the conditional quantum entropy $H\left(  A|B\right)  _{\sigma}$\ is
defined for a state $\sigma_{AB}$ as%
\begin{equation}
H\left(  A|B\right)  _{\sigma}\equiv H\left(  AB\right)  _{\sigma}-H\left(
B\right)  _{\sigma}=-D\left(  \sigma_{AB}\Vert I_{A}\otimes\sigma_{B}\right)
.
\end{equation}

Let $\omega_{XAB}$ denote the following classical-quantum state in which we
have encoded the ensemble$~\mathcal{E}$:%
\begin{equation}
\omega_{XAB}\equiv\sum_{x}p_{X}\left(  x\right)  \left\vert x\right\rangle
\left\langle x\right\vert _{X}\otimes\rho_{AB}^{x}.
\end{equation}
We can rewrite%
\begin{align}
H\left(  A|B\right)  _{\overline{\rho}}-\sum_{x}p_{X}\left(  x\right)
H\left(  A|B\right)  _{\rho^{x}} &  =H\left(  A|B\right)  _{\omega}-H\left(
A|BX\right)  _{\omega}\\
&  =I\left(  A;X|B\right)  _{\omega}\\
&  =H\left(  X|B\right)  _{\omega}-H\left(  X|AB\right)  _{\omega}\\
&  =D\left(  \omega_{XAB}\Vert I_{X}\otimes\omega_{AB}\right)  -D\left(
\omega_{XB}\Vert I_{X}\otimes\omega_{B}\right)  .
\end{align}
We can see the last line above as a relative entropy difference (as defined in
the right-hand side of \eqref{eq:rel-ent-diff-a-1}) by picking%
\begin{equation}
\rho=\omega_{XAB},\ \ \ \ \sigma=I_{X}\otimes\omega_{AB},\ \ \ \ \mathcal{N}%
=\operatorname{Tr}_{A}.
\end{equation}
Applying Theorem~\ref{thm:rel-ent}, \eqref{eq:fid-flags}, and
\eqref{eq:max-flags}, we find the following improvement of concavity of
conditional entropy:

\begin{corollary}
Let an ensemble $\mathcal{E}$\ be as given above. Then the following
inequality holds%
\begin{equation}
-2\log\sup_{t\in\mathbb{R}}\sum_{x}p_{X}\left(  x\right)  \sqrt{F}\left(
\rho_{AB}^{x},\mathcal{R}_{\overline{\rho}_{AB},\operatorname{Tr}_{A}}%
^{P,t}\left(  \rho_{B}^{x}\right)  \right)  \leq H\left(  A|B\right)
_{\overline{\rho}}-\sum_{x}p_{X}\left(  x\right)  H\left(  A|B\right)
_{\rho^{x}},
\end{equation}
where the recovery map $\mathcal{R}_{\overline{\rho}_{AB},\operatorname{Tr}%
_{A}}^{P,t}$ is defined from \eqref{eq:rotated-Petz} and perfectly recovers
$\overline{\rho}_{AB}$ from $\overline{\rho}_{B}$. If the states in the
ensemble are positive definite, then the following inequality holds%
\begin{equation}
H\left(  A|B\right)  _{\overline{\rho}}-\sum_{x}p_{X}\left(  x\right)
H\left(  A|B\right)  _{\rho^{x}}\leq\sup_{t\in\mathbb{R}}\max_{x}D_{\max
}\left(  \rho_{AB}^{x}\middle\Vert\mathcal{R}_{\overline{\rho}_{AB}%
,\operatorname{Tr}_{A}}^{P,t}\left(  \rho_{B}^{x}\right)  \right)  .
\end{equation}

\end{corollary}

\subsection{Joint convexity of quantum relative entropy}

The ideas in this section follow a line of thought developed in \cite{SBW14}.
Let $\left\{  p_{X}\left(  x\right)  ,\rho_{x}\right\}  $ be an ensemble of
density operators and $\left\{  p_{X}\left(  x\right)  ,\sigma_{x}\right\}  $
be an ensemble of positive semi-definite operators with expectations%
\begin{equation}
\overline{\rho}\equiv\sum_{x}p_{X}\left(  x\right)  \rho_{x}%
,\ \ \ \ \ \ \ \ \overline{\sigma}\equiv\sum_{x}p_{X}\left(  x\right)
\sigma_{x}.
\end{equation}
Joint convexity of quantum relative entropy is the statement that
distinguishability of these ensembles does not increase under the loss of the
classical label:%
\begin{equation}
\sum_{x}p_{X}\left(  x\right)  D\left(  \rho_{x}\Vert\sigma_{x}\right)  \geq
D\left(  \overline{\rho}\Vert\overline{\sigma}\right)  .
\end{equation}
By picking%
\begin{equation}
\rho=\rho_{XB}\equiv\sum_{x}p_{X}\left(  x\right)  \left\vert x\right\rangle
\left\langle x\right\vert _{X}\otimes\rho_{x},\ \ \ \ \sigma=\sigma_{XB}%
\equiv\sum_{x}p_{X}\left(  x\right)  \left\vert x\right\rangle \left\langle
x\right\vert _{X}\otimes\sigma_{x},\ \ \ \ \mathcal{N}=\operatorname{Tr}_{X},
\end{equation}
and applying Theorem~\ref{thm:rel-ent}, we arrive at the following improvement
of joint convexity of quantum relative entropy:

\begin{corollary}
\label{cor:JC-rel-ent} Let ensembles be as given above. Then the following
inequalities hold%
\begin{equation}
-\log\sup_{t\in\mathbb{R}}F\left(  \rho_{XB},\mathcal{R}_{\sigma
_{XB},\operatorname{Tr}_{X}}^{P,t}\left(  \overline{\rho}\right)  \right)
\leq\sum_{x}p_{X}\left(  x\right)  D\left(  \rho_{x}\Vert\sigma_{x}\right)
-D\left(  \overline{\rho}\Vert\overline{\sigma}\right)  ,
\end{equation}
where the recovery map $\mathcal{R}_{\sigma_{XB},\operatorname{Tr}_{X}}^{P,t}$
is defined from \eqref{eq:rotated-Petz} and perfectly recovers $\sigma_{XB}$
from $\sigma_{B}$. If the operators are positive definite, then%
\begin{equation}
\sum_{x}p_{X}\left(  x\right)  D\left(  \rho_{x}\Vert\sigma_{x}\right)
-D\left(  \overline{\rho}\Vert\overline{\sigma}\right)  \leq\sup
_{t\in\mathbb{R}}D_{\max}\left(  \rho_{XB}\middle\Vert\mathcal{R}_{\sigma
_{XB},\operatorname{Tr}_{X}}^{P,t}\left(  \overline{\rho}\right)  \right)  .
\end{equation}

\end{corollary}

\begin{remark}
The recovery map $\mathcal{R}_{\sigma_{XB},\operatorname{Tr}_{X}}^{P,t}$ from
Corollary~\ref{cor:JC-rel-ent} can be understood as a rotated
\textquotedblleft pretty good\textquotedblright\ measurement
\cite{B75,B75a,PhysRevA.54.1869}. By using the facts that%
\begin{equation}
\rho_{XB}=\bigoplus\limits_{x}p_{X}\left(  x\right)  \rho_{x},\ \ \ \ \sigma
_{XB}=\bigoplus\limits_{x}p_{X}\left(  x\right)  \sigma_{x},
\end{equation}
we have for all $z\in\mathbb{C}$ that%
\begin{equation}
\rho_{XB}^{z}=\bigoplus\limits_{x}\left[  p_{X}\left(  x\right)  \rho
_{x}\right]  ^{z},\ \ \ \ \sigma_{XB}^{z}=\bigoplus\limits_{x}\left[
p_{X}\left(  x\right)  \sigma_{x}\right]  ^{z}.
\end{equation}
This allows us to write the recovery map as the following trace-non-increasing
instrument (a trace-non-increasing map with a classical and quantum output):%
\begin{equation}
\mathcal{R}_{\sigma_{XB},\operatorname{Tr}_{X}}^{P,t}\left(  \cdot\right)
=\sum_{x}\left\vert x\right\rangle \left\langle x\right\vert _{X}\otimes
p_{X}\left(  x\right)  \left[  p_{X}\left(  x\right)  \sigma_{x}\right]
^{it}\sigma_{x}^{1/2}\left(  \overline{\sigma}\right)  ^{-\frac{1}{2}%
}(\overline{\sigma})^{-it}\left(  \cdot\right)  (\overline{\sigma}%
)^{it}\left(  \overline{\sigma}\right)  ^{-\frac{1}{2}}\sigma_{x}^{1/2}\left[
p_{X}\left(  x\right)  \sigma_{x}\right]  ^{-it}.
\end{equation}
Tracing over the quantum system then gives the following rotated pretty good
measurement:%
\begin{equation}
\left(  \cdot\right)  \rightarrow\sum_{x}\operatorname{Tr}\left\{  \left(
\overline{\sigma}\right)  ^{-1/2}p_{X}\left(  x\right)  \sigma_{x}\left(
\overline{\sigma}\right)  ^{-1/2}\overline{\sigma}^{-it}\left(  \cdot\right)
\overline{\sigma}^{it}\right\}  \left\vert x\right\rangle \left\langle
x\right\vert _{X},
\end{equation}
which should be compared with the measurement map corresponding to the pretty
good measurement:%
\begin{equation}
\left(  \cdot\right)  \rightarrow\sum_{x}\operatorname{Tr}\left\{  \left(
\overline{\sigma}\right)  ^{-1/2}p_{X}\left(  x\right)  \sigma_{x}\left(
\overline{\sigma}\right)  ^{-1/2}\left(  \cdot\right)  \right\}  \left\vert
x\right\rangle \left\langle x\right\vert _{X}.
\end{equation}

\end{remark}

\subsection{Non-negativity of quantum discord}

The ideas in this section follow a line of thought developed in
\cite{SBW14,SW14}. Let $\rho_{AB}$ be a bipartite density operator and let
$\left\{  \left\vert \varphi_{x}\right\rangle \left\langle \varphi
_{x}\right\vert _{A}\right\}  $ be a rank-one quantum measurement on system
$A$ (i.e., the vectors $\left\vert \varphi_{x}\right\rangle _{A}$ satisfy
$\sum_{x}\left\vert \varphi_{x}\right\rangle \left\langle \varphi
_{x}\right\vert _{A}=I_{A}$). It suffices for us to consider rank-one
measurements for our discussion here because every quantum measurement can be
refined to have a rank-one form, such that it delivers more classical
information to the experimentalist observing the apparatus. Then the
(unoptimized) quantum discord is defined to be the difference between the
following mutual informations \cite{Z00,zurek}:%
\begin{equation}
I\left(  A;B\right)  _{\rho}-I\left(  X;B\right)  _{\omega},
\end{equation}
where%
\begin{align}
I\left(  A;B\right)  _{\rho} &  \equiv D\left(  \rho_{AB}\Vert\rho_{A}%
\otimes\rho_{B}\right)  ,\\
\mathcal{M}_{A\rightarrow X}\left(  \cdot\right)   &  \equiv\sum
_{x}\left\langle \varphi_{x}\right\vert _{A}\left(  \cdot\right)  \left\vert
\varphi_{x}\right\rangle _{A}\left\vert x\right\rangle \left\langle
x\right\vert _{X},\label{eq:meas-discord}\\
\omega_{XB} &  \equiv\mathcal{M}_{A\rightarrow X}\left(  \rho_{AB}\right)
.\label{eq:post-meas-discord}%
\end{align}
The quantum channel $\mathcal{M}_{A\rightarrow X}$ is a measurement channel,
so that the state $\omega_{XB}$ is the classical-quantum state resulting from
the measurement. The set $\left\{  \left\vert x\right\rangle _{X}\right\}  $
is an orthonormal basis so that $X$ is a classical system. The quantum discord
is known to be non-negative, and by applying Theorem~\ref{thm:rel-ent}\ we
find the following improvement of this entropy inequality:

\begin{corollary}
\label{cor:discord}Let $\rho_{AB}$ and $\mathcal{M}_{A\rightarrow X}$ be as
given above. Then the following inequalities hold%
\begin{equation}
-\log\sup_{t\in\mathbb{R}}F\left(  \rho_{AB},\left(  \mathcal{U}_{\rho_{A}%
,t}\circ\mathcal{E}_{A}\right)  \left(  \rho_{AB}\right)  \right)  \leq
I\left(  A;B\right)  _{\rho}-I\left(  X;B\right)  _{\omega}%
,\label{eq:discord-bound}%
\end{equation}
where $\mathcal{E}_{A}$ is an entanglement-breaking map of the following form:%
\begin{equation}
\left(  \cdot\right)  _{A}\rightarrow\sum_{x}\left\langle \varphi
_{x}\right\vert _{A}\left(  \cdot\right)  \left\vert \varphi_{x}\right\rangle
_{A}\frac{\rho_{A}^{1/2}\left\vert \varphi_{x}\right\rangle \left\langle
\varphi_{x}\right\vert _{A}\rho_{A}^{1/2}}{\left\langle \varphi_{x}\right\vert
_{A}\rho_{A}\left\vert \varphi_{x}\right\rangle _{A}},\label{eq:EB-channel}%
\end{equation}
and the partial isometric map $\mathcal{U}_{\rho_{A},t}$ is defined from
\eqref{eq:unitaries}. The recovery map $\mathcal{U}_{\rho_{A},t}%
\circ\mathcal{E}_{A}$ perfectly recovers $\rho_{A}$ from $\mathcal{M}%
_{A\rightarrow X}(\rho_{A})$.
\end{corollary}

\begin{proof}
We start with the rewriting%
\begin{equation}
I\left(  A;B\right)  _{\rho}-I\left(  X;B\right)  _{\omega}=D\left(  \rho
_{AB}\Vert\rho_{A}\otimes I_{B}\right)  -D\left(  \omega_{AB}\Vert\omega
_{A}\otimes I_{B}\right)  ,
\end{equation}
and follow by picking%
\begin{equation}
\rho=\rho_{AB},\ \ \ \ \sigma=\rho_{A}\otimes I_{B},\ \ \ \ \mathcal{N}%
=\mathcal{M}_{A\rightarrow X},
\end{equation}
and applying Theorem~\ref{thm:rel-ent}. This then shows the corollary with a
recovery map of the form $\mathcal{R}_{\rho_{A},\mathcal{M}_{A\rightarrow X}%
}^{P,t}\circ\mathcal{M}_{A\rightarrow X}$.

As observed in \cite{SBW14,SW14}, the concatenation $\mathcal{R}_{\rho
_{A},\mathcal{M}_{A\rightarrow X}}^{P,t}\circ\mathcal{M}_{A\rightarrow X}$ is
an entanglement-breaking channel \cite{HSR03} because it consists of a
measurement channel $\mathcal{M}_{A\rightarrow X}$ followed by a preparation.
We now work out the form for the recovery map given in
\eqref{eq:discord-bound}. Consider that%
\begin{equation}
\mathcal{M}_{A\rightarrow X}\left(  \rho_{A}\right)  =\sum_{x}\left\langle
\varphi_{x}\right\vert _{A}\rho_{A}\left\vert \varphi_{x}\right\rangle
_{A}\left\vert x\right\rangle \left\langle x\right\vert _{X},
\end{equation}
so that%
\begin{equation}
\mathcal{U}_{\mathcal{M}_{A\rightarrow X}\left(  \rho_{A}\right)  ,-t}\left(
\cdot\right)  =\left[  \sum_{x}\left[  \left\langle \varphi_{x}\right\vert
_{A}\rho_{A}\left\vert \varphi_{x}\right\rangle _{A}\right]  ^{-it}\left\vert
x\right\rangle \left\langle x\right\vert _{X}\right]  \left(  \cdot\right)
\left[  \sum_{x^{\prime}}\left[  \left\langle \varphi_{x^{\prime}}\right\vert
_{A}\rho_{A}\left\vert \varphi_{x^{\prime}}\right\rangle _{A}\right]
^{it}\left\vert x^{\prime}\right\rangle \left\langle x^{\prime}\right\vert
_{X}\right]  .
\end{equation}
Thus, when composing $\mathcal{M}_{A\rightarrow X}$ with $\mathcal{U}%
_{\mathcal{M}_{A\rightarrow X}\left(  \rho_{A}\right)  ,-t}$, the phases
cancel out to give the following relation:%
\begin{equation}
\mathcal{U}_{\mathcal{M}_{A\rightarrow X}\left(  \rho_{A}\right)  ,-t}\left(
\mathcal{M}_{A\rightarrow X}\left(  \cdot\right)  \right)  =\mathcal{M}%
_{A\rightarrow X}\left(  \cdot\right)  .
\end{equation}
One can then work out that%
\begin{align}
\mathcal{R}_{\rho_{A},\mathcal{M}_{A\rightarrow X}}^{P}\left(  \mathcal{M}%
_{A\rightarrow X}\left(  \cdot\right)  \right)   &  =\rho_{A}^{1/2}%
\mathcal{M}^{\dag}\left(  \left[  \mathcal{M}_{A\rightarrow X}\left(  \rho
_{A}\right)  \right]  ^{-1/2}\mathcal{M}_{A\rightarrow X}\left(  \cdot\right)
\left[  \mathcal{M}_{A\rightarrow X}\left(  \rho_{A}\right)  \right]
^{-1/2}\right)  \rho_{A}^{1/2}\\
&  =\sum_{x}\left\langle \varphi_{x}\right\vert _{A}\left(  \cdot\right)
\left\vert \varphi_{x}\right\rangle _{A}\frac{\rho_{A}^{1/2}\left\vert
\varphi_{x}\right\rangle \left\langle \varphi_{x}\right\vert _{A}\rho
_{A}^{1/2}}{\left\langle \varphi_{x}\right\vert _{A}\rho_{A}\left\vert
\varphi_{x}\right\rangle _{A}}.
\end{align}

\end{proof}

\subsection{Holevo bound}

The ideas in this section follow a line of thought developed in \cite{SBW14}.
The Holevo bound \cite{Ho73}\ is a special case of the non-negativity of
quantum discord in which $\rho_{AB}$ is a quantum-classical state, which we
write explicitly as%
\begin{equation}
\rho_{AB}=\sum_{y}p_{Y}\left(  y\right)  \rho_{A}^{y}\otimes\left\vert
y\right\rangle \left\langle y\right\vert _{Y},\label{eq:cq-holevo}%
\end{equation}
where each $\rho_{A}^{y}$ is a density operator, so that $\rho_{A}=\sum
_{y}p_{Y}\left(  y\right)  \rho_{A}^{y}$. The Holevo bound states that the
mutual information of the state $\rho_{AB}$ in \eqref{eq:cq-holevo} is never
smaller than the mutual information after system $A$ is measured. By applying
Corollary~\ref{cor:discord} and \eqref{eq:fid-flags}, we find the following improvement:

\begin{corollary}
Let $\rho_{AB}$ be as in \eqref{eq:cq-holevo},\ and let $\mathcal{M}%
_{A\rightarrow X}$ and $\omega_{AB}$ be as in
\eqref{eq:meas-discord}-\eqref{eq:post-meas-discord}, respectively. Then the
following inequality holds%
\begin{equation}
-2\log\sup_{t\in\mathbb{R}}\sum_{y}p_{Y}\left(  y\right)  \sqrt{F}\left(
\rho_{A}^{y},\left(  \mathcal{U}_{\rho_{A},t}\circ\mathcal{E}_{A}\right)
\left(  \rho_{A}^{y}\right)  \right)  \leq I\left(  A;B\right)  _{\rho
}-I\left(  X;B\right)  _{\omega},
\end{equation}
where $\mathcal{E}_{A}$ is an entanglement-breaking map of the form in
\eqref{eq:EB-channel} and the partial isometric map $\mathcal{U}_{\rho_{A},t}$
is defined from \eqref{eq:unitaries}.
\end{corollary}

\subsection{Differences of quantum multipartite informations}

The quantum multipartite information of a multipartite state $\omega
_{B_{1}\cdots B_{l}}$ is defined as \cite{W60,H94}:%
\begin{equation}
I\left(  B_{1}:\cdots:B_{l}\right)  _{\omega}\equiv\sum_{i=1}^{l}H\left(
B_{i}\right)  _{\omega}-H\left(  B_{1}\cdots B_{l}\right)  _{\omega}=D\left(
\omega_{B_{1}\cdots B_{l}}\Vert\omega_{B_{1}}\otimes\cdots\otimes\omega
_{B_{l}}\right)  .
\end{equation}
One can then use this to define a difference of multipartite informations for
a state $\rho_{A_{1}A_{1}^{\prime}\cdots A_{l}A_{l}^{\prime}}$\ as%
\begin{equation}
I\left(  A_{1}A_{1}^{\prime}:\cdots:A_{l}A_{l}^{\prime}\right)  _{\rho
}-I\left(  A_{1}^{\prime}:\cdots:A_{l}^{\prime}\right)  _{\rho},
\end{equation}
using which multipartite entanglement \cite{YHW08}\ and discord-like
\cite{PHH08} measures can be constructed. By picking%
\begin{equation}
\rho=\rho_{A_{1}A_{1}^{\prime}\cdots A_{l}A_{l}^{\prime}},\ \ \ \ \sigma
=\rho_{A_{1}A_{1}^{\prime}}\otimes\cdots\otimes\rho_{A_{l}A_{l}^{\prime}%
},\ \ \ \ \mathcal{N}=\operatorname{Tr}_{A_{1}\cdots A_{l}},
\end{equation}
and applying Theorem~\ref{thm:rel-ent}, we establish the following corollary,
which solves in the affirmative the open question posed in Eq.~(7.8) of
\cite{W14}:

\begin{corollary}
Let $\rho_{A_{1}A_{1}^{\prime}\cdots A_{l}A_{l}^{\prime}}$ be a multipartite
quantum state. Then the following inequality holds%
\begin{multline}
-\log\sup_{t\in\mathbb{R}}F\left(  \rho_{A_{1}A_{1}^{\prime}\cdots A_{l}%
A_{l}^{\prime}},\left(  \mathcal{R}_{\rho_{A_{1}A_{1}^{\prime}}%
,\operatorname{Tr}_{A_{1}}}^{P,t}\otimes\cdots\otimes\mathcal{R}_{\rho
_{A_{l}A_{l}^{\prime}},\operatorname{Tr}_{A_{l}}}^{P,t}\right)  \left(
\rho_{A_{1}^{\prime}\cdots A_{l}^{\prime}}\right)  \right)  \\
\leq I\left(  A_{1}A_{1}^{\prime}:\cdots:A_{l}A_{l}^{\prime}\right)  _{\rho
}-I\left(  A_{1}^{\prime}:\cdots:A_{l}^{\prime}\right)  _{\rho}.
\end{multline}
If the state is positive definite, then the following inequality holds as
well:%
\begin{multline}
I\left(  A_{1}A_{1}^{\prime}:\cdots:A_{l}A_{l}^{\prime}\right)  _{\rho
}-I\left(  A_{1}^{\prime}:\cdots:A_{l}^{\prime}\right)  _{\rho}\\
\leq\sup_{t\in\mathbb{R}}D_{\max}\left(  \rho_{A_{1}A_{1}^{\prime}\cdots
A_{l}A_{l}^{\prime}}\middle\Vert\left(  \mathcal{R}_{\rho_{A_{1}A_{1}^{\prime
}},\operatorname{Tr}_{A_{1}}}^{P,t}\otimes\cdots\otimes\mathcal{R}%
_{\rho_{A_{l}A_{l}^{\prime}},\operatorname{Tr}_{A_{l}}}^{P,t}\right)  \left(
\rho_{A_{1}^{\prime}\cdots A_{l}^{\prime}}\right)  \right)  .
\end{multline}

\end{corollary}

\subsection{General quantum information measures}

We remark here that the method given in the proof of Theorem~\ref{thm:rel-ent}
can be applied quite generally, even to information quantities which cannot be
written as a difference of relative entropies. To do so, one needs to follow
the recipe outlined in \cite{BSW15a}\ for obtaining a R\'{e}nyi generalization
of the entropic quantity of interest and then apply the same methods of
complex interpolation used in the proof of Theorem~\ref{thm:rel-ent}. The
bounds that one ends up with might not necessarily have a physical
interpretation in terms of recoverability, but it does happen in some cases.

One example of an information quantity for which this does happen and for
which it is not clear how to write it in terms of a relative entropy
difference is the conditional multipartite information of a state $\rho
_{A_{1}\cdots A_{l}C}$:%
\begin{equation}
I\left(  A_{1}:\cdots:A_{l}|C\right)  _{\rho}\equiv\sum_{i=1}^{l}H\left(
A_{i}|C\right)  _{\rho}-H\left(  A_{1}\cdots A_{l}|C\right)  _{\rho}.
\end{equation}
This quantity can be used to define squashed-like entanglement measures
\cite{YHHHOS09,AHS08}. It is not clear how to write this as a relative entropy
difference, but one can follow the recipe given in \cite{BSW15a} to find the
following R\'{e}nyi generalization:%
\begin{equation}
\widetilde{I}_{\alpha}\left(  A_{1}:\cdots:A_{l}|C\right)  _{\rho}\equiv
\frac{2}{\alpha^{\prime}}\log\left\Vert \rho_{A_{1}\cdots A_{l}C}^{1/2}%
\rho_{A_{l}C}^{-\alpha^{\prime}/2}\rho_{C}^{\alpha^{\prime}/2}\cdots
\rho_{A_{2}C}^{-\alpha^{\prime}/2}\rho_{C}^{\alpha^{\prime}/2}\rho_{A_{1}%
C}^{-\alpha^{\prime}/2}\right\Vert _{2\alpha},
\end{equation}
where $\alpha^{\prime}=\left(  \alpha-1\right)  /\alpha$. The quantity
$\widetilde{I}_{\alpha}\left(  A_{1}:\cdots:A_{l}|C\right)  $ converges to
$I\left(  A_{1}:\cdots:A_{l}|C\right)  $ in the limit as $\alpha\rightarrow
1$.\ For $\alpha=1/2$ and in the limit as $\alpha\rightarrow\infty$,
$\widetilde{I}_{\alpha}\left(  A_{1}:\cdots:A_{l}|C\right)  $ reduces to%
\begin{align}
\widetilde{I}_{1/2}\left(  A_{1}:\cdots:A_{l}|C\right)  _{\rho} &  =-\log
F\left(  \rho_{A_{1}\cdots A_{l}C},\left(  \mathcal{R}_{C\rightarrow A_{l}%
C}^{P}\circ\cdots\circ\mathcal{R}_{C\rightarrow A_{2}C}^{P}\right)  \left(
\rho_{A_{1}C}\right)  \right)  ,\\
\widetilde{I}_{\infty}\left(  A_{1}:\cdots:A_{l}|C\right)  _{\rho} &
=D_{\max}\left(  \rho_{A_{1}\cdots A_{l}C}\middle\Vert\left(  \mathcal{R}%
_{C\rightarrow A_{l}C}^{P}\circ\cdots\circ\mathcal{R}_{C\rightarrow A_{2}%
C}^{P}\right)  \left(  \rho_{A_{1}C}\right)  \right)  ,
\end{align}
where $\mathcal{R}_{C\rightarrow A_{i}C}^{P}$ is a Petz recovery map of the
form $\left(  \cdot\right)  \rightarrow\rho_{A_{i}C}^{1/2}\rho_{C}%
^{-1/2}\left(  \cdot\right)  \rho_{C}^{-1/2}\rho_{A_{i}C}^{1/2}$ and the
latter equality requires the state to be positive definite. The quantities
above have an interpretation in terms of \textit{sequential recoverability},
that is where one attempts to use the system $C$ repeatedly in order to
retrieve all of the $A_{i}$ systems one-by-one, for $i\in\left\{
2,\ldots,l\right\}  $. One can exploit the method of proof for
Theorem~\ref{thm:rel-ent} to obtain the following bounds:

\begin{theorem}
Let $\rho_{A_{1}\cdots A_{l}C}$ be a density operator. Then the following
inequalities hold%
\begin{equation}
-\log\sup_{t\in\mathbb{R}}F\left(  \rho_{A_{1}\cdots A_{l}C},\left(
\mathcal{R}_{C\rightarrow A_{l}C}^{P,t}\circ\cdots\circ\mathcal{R}%
_{C\rightarrow A_{2}C}^{P,t}\right)  \left(  \rho_{A_{1}C}\right)  \right)
\leq I\left(  A_{1}:\cdots:A_{l}|C\right)  _{\rho},
\end{equation}
where $\mathcal{R}_{C\rightarrow A_{i}C}^{P,t}$ is a rotated Petz recovery map
of the form in \eqref{eq:CMI-petz-recovery}. If the state is positive
definite, then%
\begin{equation}
I\left(  A_{1}:\cdots:A_{l}|C\right)  _{\rho}\leq\sup_{t\in\mathbb{R}}D_{\max
}\left(  \rho_{A_{1}\cdots A_{l}C}\middle\Vert\left(  \mathcal{R}%
_{C\rightarrow A_{l}C}^{P,t}\circ\cdots\circ\mathcal{R}_{C\rightarrow A_{2}%
C}^{P,t}\right)  \left(  \rho_{A_{1}C}\right)  \right)  .
\end{equation}

\end{theorem}

\begin{remark}
An advantage of the bound given above is that there is no prefactor depending
on the number of parties involved, which one obtains by applying the chain
rule and the bound from \cite{FR14} multiple times (cf., Appendix B of
\cite{LW14}). Furthermore, it is not known how to apply the methods of
\cite{FR14,SOR15} in order to arrive at a bound of the form above (which does
not have a prefactor depending on the number of parties).
\end{remark}

\subsection{Approximate quantum error correction}

The goal of quantum error correction is to protect quantum information from
the deleterious effects of a quantum channel $\mathcal{N}$ by encoding it into
a subspace of the full Hilbert space, such that one can later recover the
encoded data after performing a recovery operation. In physical situations,
one can never have perfect error correction and instead aims for approximate
error correction (see, e.g., \cite{Mandayam2012} and references therein). In
more detail, let $\mathcal{H}$ be a Hilbert space and let $\Pi$ be a
projection onto some subspace of $\mathcal{H}$, which is referred to as the
codespace. Suppose that $\rho$ is a density operator with support only in the
subspace onto which $\Pi$ projects. Then, by choosing%
\begin{equation}
\rho=\rho,\ \ \ \ \sigma=\Pi,\ \ \ \ \mathcal{N}=\mathcal{N},
\end{equation}
and applying Theorem~\ref{thm:rel-ent}, we find that the following inequality
holds%
\begin{equation}
-\log\left[  \sup_{t\in\mathbb{R}}F\left(  \rho,\mathcal{R}_{\Pi,\mathcal{N}%
}^{P,t}\left(  \mathcal{N(}\rho)\right)  \right)  \right]  \leq D\left(
\rho\Vert\Pi\right)  -D\left(  \mathcal{N(}\rho)\Vert\mathcal{N(}\Pi)\right)
,
\end{equation}
where $\mathcal{R}_{\Pi,\mathcal{N}}^{P,t}$ is the following rotated Petz
recovery map:%
\begin{equation}
\mathcal{R}_{\Pi,\mathcal{N}}^{P,t}\left(  \cdot\right)  \equiv\left(
\mathcal{U}_{\Pi,t}\circ\mathcal{R}_{\Pi,\mathcal{N}}^{P}\circ\mathcal{U}%
_{\mathcal{N(}\Pi),-t}\right)  \left(  \cdot\right)  ,
\end{equation}
$\mathcal{R}_{\Pi,\mathcal{N}}^{P}$ is the Petz recovery or \textquotedblleft
transpose\textquotedblright\ map defined as%
\begin{equation}
\mathcal{R}_{\Pi,\mathcal{N}}^{P}\left(  \cdot\right)  =\Pi\mathcal{N}^{\dag
}\left(  \left[  \mathcal{N(}\Pi)\right]  ^{-1/2}\left(  \cdot\right)  \left[
\mathcal{N(}\Pi)\right]  ^{-1/2}\right)  \Pi,
\end{equation}
and $\mathcal{U}_{\Pi,t}$ and $\mathcal{U}_{\mathcal{N(}\Pi),-t}$ are partial
isometric maps defined from \eqref{eq:unitaries} and acting as the identity
outside the support of $\Pi$ and $\mathcal{N(}\Pi)$, respectively.

The inequality above is not particularly useful in the context of approximate
quantum error correction, but we have stated it to motivate further
developments. In particular, the recovery map $\mathcal{R}_{\Pi,\mathcal{N}%
}^{P,t}$ has a dependence on the particular state $\rho$ being sent through
the channel. Of course, the receiver does not know which state is being
transmitted through the channel at any given instant and thus cannot apply the
decoder given in the above bound. Nor should we allow the encoder to know
which state $\rho$ is being transmitted and send $t$ via a noiseless classical
channel. Thus, it would be ideal if the bound stated above would hold for a
recovery map that has no dependence on the input (that is, if the recovery map
were to have a universality property, similar to that discussed in
\cite{SOR15} and Remark~\ref{rem:perfect-recovery}).\ Some possibilities for
universal recovery maps are the Petz recovery map itself or an averaged
channel where $t$ is chosen randomly according to some distribution.
Conjecture~26 of \cite{SBW14}, if true, would imply the bound given above for
the choice $t=0$ (the Petz recovery map).

\section{Discussion}

\label{sec:discuss}Entropy inequalities such as strong subadditivity of
quantum entropy \cite{LR73,PhysRevLett.30.434} and monotonicity of quantum
relative entropy \cite{Lindblad1975,U77} have played a fundamental role in
quantum information theory and other areas of physics. Establishing entropy
inequalities with physically meaningful remainder terms has been a topic of
recent interest in quantum information theory (see \cite{LW14,FR14,SOR15} and
references therein). A breakthrough result from \cite{FR14} established an
inequality of the form in Theorem~\ref{thm:CMI}, with however essentially
nothing being known about the input and output unitaries. The methods of
\cite{FR14} were generalized in \cite{BLW14} to produce an inequality of the
form in \eqref{eq:rel-ent-ineq}, again with essentially nothing known about
the input and output unitaries. Meanwhile, operational proofs for physically
meaningful lower bounds on conditional mutual information have appeared as
well \cite{BHOS14,BT15}, the latter in part based on the notion of fidelity of
recovery \cite{SW14}. Recent work has now established that a recovery map for
the conditional mutual information can possess a universality property
\cite{SOR15}, in the sense that it need not depend on the state of the system
$B$ (the system that is not \textquotedblleft lost and
recovered\textquotedblright\ nor \textquotedblleft used to
recover\textquotedblright). As discussed in Remark~\ref{rem:perfect-recovery},
the recovery map given in \eqref{eq:main-result} does not possess the
universality property. Also, the structure of the input and output unitaries
given in \eqref{eq:main-result} is essentially the same as that found in
previous work \cite{FR14,SOR15} (see the discussion around Eqs.~(120)-(122) of
\cite{SOR15}). However, the argument given here to arrive at this structure is
more direct than that in \cite{FR14,SOR15} and applies as well to the recovery
map in \eqref{eq:rel-ent-ineq} for a relative entropy difference.

An important open question is to determine whether we could take $t=0$ and
still have the inequalities hold, as conjectured previously
\cite{BSW14,SBW14,LW14}. More generally, it is still open to determine whether
the R\'{e}nyi quantities in \eqref{eq:renyi-diff} and \eqref{eq:Renyi-CMI} are
monotone non-decreasing with respect to the R\'{e}nyi parameter $\alpha$.

In light of the efforts put into addressing the recovery question, it is
pleasing that the Hadamard three-line theorem leads to simple proofs. This
theorem has already been put to good use in characterizing local state
transformations \cite{DB15}\ and in obtaining chain rules for R\'{e}nyi
entropies \cite{D14}, for example, and it should be interesting to find
further applications of it in the context of quantum information theory.

\bigskip

\textbf{Acknowledgements.} I am especially grateful to my collaborators Mario
Berta and Kaushik P. Seshadreesan for our many joint contributions on this
topic \cite{BSW14,SBW14,BSW15a,BLW14,SW14}, upon which the present work
builds. I acknowledge many discussions with and/or feedback on the manuscript
from Mario Berta, Tom Cooney, Nilanjana Datta, Frederic Dupuis, Omar Fawzi,
Rupert Frank, Milan Mosonyi, Renato Renner, Volkher Scholz, Kaushik
Seshadreesan, David Sutter, Marco Tomamichel, Stephanie Wehner, and Andreas
Winter. I am grateful to Stephanie Wehner and her group for hospitality during
a research visit to TU\ Delft (May 2015) and acknowledge support from startup
funds from the Department of Physics and Astronomy at LSU, the NSF\ under
Award No.~CCF-1350397, and the DARPA Quiness Program through US Army Research
Office award W31P4Q-12-1-0019.

\appendix

\section{Convergence to the quantum relative entropy difference}

\label{app:limit}

\begin{definition}
Let $\rho$, $\sigma$, and $\mathcal{N}$ be as given in
Definition~\ref{def:rho-sig-N}. For $\alpha\in\left(  0,1\right)  \cup\left(
1,\infty\right)  $, let%
\begin{equation}
\widetilde{\Delta}_{\alpha}\left(  \rho,\sigma,\mathcal{N}\right)  =\frac
{1}{\alpha-1}\log\widetilde{Q}_{\alpha}\left(  \rho,\sigma,\mathcal{N}\right)
,
\end{equation}
where%
\begin{equation}
\widetilde{Q}_{\alpha}\left(  \rho,\sigma,\mathcal{N}\right)  \equiv\left\Vert
\left(  \mathcal{N}\left(  \rho\right)  ^{\left(  1-\alpha\right)  /2\alpha
}\mathcal{N}\left(  \sigma\right)  ^{\left(  \alpha-1\right)  /2\alpha}\otimes
I_{E}\right)  U\sigma^{\left(  1-\alpha\right)  /2\alpha}\rho^{1/2}\right\Vert
_{2\alpha}^{2\alpha}.
\end{equation}

\end{definition}

\begin{theorem}
\label{thm:tilde-alpha-limit}Let $\rho$, $\sigma$, and $\mathcal{N}$ be as
given in Definition~\ref{def:rho-sig-N}. The following limit holds%
\begin{equation}
\lim_{\alpha\rightarrow1}\widetilde{\Delta}_{\alpha}\left(  \rho
,\sigma,\mathcal{N}\right)  =D\left(  \rho\Vert\sigma\right)  -D\left(
\mathcal{N}\left(  \rho\right)  \Vert\mathcal{N}\left(  \sigma\right)
\right)  .
\end{equation}

\end{theorem}

\begin{proof}
Let $\Pi_{\omega}$ denote the projection onto the support of $\omega$. From
the condition $\operatorname{supp}\left(  \rho\right)  \subseteq
\operatorname{supp}\left(  \sigma\right)  $, it follows that
$\operatorname{supp}\left(  \mathcal{N}\left(  \rho\right)  \right)
\subseteq\operatorname{supp}\left(  \mathcal{N}\left(  \sigma\right)  \right)
$ \cite[Appendix~B.4]{R05}. We can then conclude that%
\begin{equation}
\Pi_{\sigma}\Pi_{\rho}=\Pi_{\rho},\ \ \ \ \ \ \ \ \Pi_{\mathcal{N}\left(
\rho\right)  }\Pi_{\mathcal{N}\left(  \sigma\right)  }=\Pi_{\mathcal{N}\left(
\rho\right)  }.\label{eq:Pi-rho}%
\end{equation}
We also know that $\operatorname{supp}\left(  U\rho U^{\dag}\right)
\subseteq\operatorname{supp}\left(  \mathcal{N}\left(  \rho\right)  \otimes
I_{E}\right)  $ \cite[Appendix~B.4]{R05}, so that%
\begin{equation}
\left(  \Pi_{\mathcal{N}\left(  \rho\right)  }\otimes I_{E}\right)  \Pi_{U\rho
U^{\dag}}=\Pi_{U\rho U^{\dag}}.\label{eq:Pi-U-rho}%
\end{equation}
When $\alpha=1$, we find from the above facts that%
\begin{align}
\widetilde{Q}_{1}\left(  \rho,\sigma,\mathcal{N}\right)   &  =\left\Vert
\left(  \Pi_{\mathcal{N}\left(  \rho\right)  }\Pi_{\mathcal{N}\left(
\sigma\right)  }\otimes I_{E}\right)  U\Pi_{\sigma}\rho^{1/2}\right\Vert
_{2}^{2}=\left\Vert \left(  \Pi_{\mathcal{N}\left(  \rho\right)  }\otimes
I_{E}\right)  U\Pi_{\rho}\rho^{1/2}\right\Vert _{2}^{2}\nonumber\\
&  =\left\Vert \left(  \Pi_{\mathcal{N}\left(  \rho\right)  }\otimes
I_{E}\right)  \Pi_{U\rho U^{\dag}}U\rho^{1/2}\right\Vert _{2}^{2}=\left\Vert
\Pi_{U\rho U^{\dag}}U\rho^{1/2}\right\Vert _{2}^{2}=\left\Vert \rho
^{1/2}\right\Vert _{2}^{2}=1.
\end{align}
So from the definition of the derivative, this means that%
\begin{align}
\lim_{\alpha\rightarrow1}\widetilde{\Delta}_{\alpha}\left(  \rho
,\sigma,\mathcal{N}\right)   &  =\lim_{\alpha\rightarrow1}\frac{\log
\widetilde{Q}_{\alpha}\left(  \rho,\sigma,\mathcal{N}\right)  -\log
\widetilde{Q}_{1}\left(  \rho,\sigma,\mathcal{N}\right)  }{\alpha-1}=\left.
\frac{d}{d\alpha}\left[  \log\widetilde{Q}_{\alpha}\left(  \rho,\sigma
,\mathcal{N}\right)  \right]  \right\vert _{\alpha=1}\nonumber\\
&  =\frac{1}{\widetilde{Q}_{1}\left(  \rho,\sigma,\mathcal{N}\right)  }\left.
\frac{d}{d\alpha}\left[  \widetilde{Q}_{\alpha}\left(  \rho,\sigma
,\mathcal{N}\right)  \right]  \right\vert _{\alpha=1}=\left.  \frac{d}%
{d\alpha}\left[  \widetilde{Q}_{\alpha}\left(  \rho,\sigma,\mathcal{N}\right)
\right]  \right\vert _{\alpha=1}.\label{eq:limit-a-1-exp}%
\end{align}
Let%
\begin{equation}
\alpha^{\prime}\equiv\frac{\alpha-1}{\alpha}.
\end{equation}
Now consider that%
\begin{equation}
\widetilde{Q}_{\alpha}\left(  \rho,\sigma,\mathcal{N}\right)
=\operatorname{Tr}\left\{  \left[  \rho^{1/2}\sigma^{-\alpha^{\prime}%
/2}\mathcal{N}^{\dag}\left(  \mathcal{N}\left(  \sigma\right)  ^{\alpha
^{\prime}/2}\mathcal{N}\left(  \rho\right)  ^{-\alpha^{\prime}}\mathcal{N}%
\left(  \sigma\right)  ^{\alpha^{\prime}/2}\right)  \sigma^{-\alpha^{\prime
}/2}\rho^{1/2}\right]  ^{\alpha}\right\}  .
\end{equation}
Define the function%
\begin{equation}
\widetilde{Q}_{\alpha,\beta}\left(  \rho,\sigma,\mathcal{N}\right)
\equiv\operatorname{Tr}\left\{  \left[  \rho^{1/2}\sigma^{-\alpha^{\prime}%
/2}\mathcal{N}^{\dag}\left(  \mathcal{N}\left(  \sigma\right)  ^{\alpha
^{\prime}/2}\mathcal{N}\left(  \rho\right)  ^{-\alpha^{\prime}}\mathcal{N}%
\left(  \sigma\right)  ^{\alpha^{\prime}/2}\right)  \sigma^{-\alpha^{\prime
}/2}\rho^{1/2}\right]  ^{\beta}\right\}  ,
\end{equation}
and consider that%
\begin{align}
\left.  \frac{d}{d\alpha}\left[  \widetilde{Q}_{\alpha}\left(  \rho
,\sigma,\mathcal{N}\right)  \right]  \right\vert _{\alpha=1} &  =\left.
\frac{d}{d\alpha}\widetilde{Q}_{\alpha,\alpha}\left(  \rho,\sigma
,\mathcal{N}\right)  \right\vert _{\alpha=1}\nonumber\\
&  =\left.  \frac{d}{d\alpha}\widetilde{Q}_{\alpha,1}\left(  \rho
,\sigma,\mathcal{N}\right)  \right\vert _{\alpha=1}+\left.  \frac{d}{d\beta
}\widetilde{Q}_{1,\beta}\left(  \rho,\sigma,\mathcal{N}\right)  \right\vert
_{\beta=1}.\label{eq:chain-break-up}%
\end{align}
We first compute $\widetilde{Q}_{1,\beta}\left(  \rho,\sigma,\mathcal{N}%
\right)  $ as follows:%
\begin{align}
\widetilde{Q}_{1,\beta}\left(  \rho,\sigma,\mathcal{N}\right)   &
=\operatorname{Tr}\left\{  \left[  \rho^{1/2}\Pi_{\sigma}\mathcal{N}^{\dag
}\left(  \Pi_{\mathcal{N}\left(  \sigma\right)  }\Pi_{\mathcal{N}\left(
\rho\right)  }\Pi_{\mathcal{N}\left(  \sigma\right)  }\right)  \Pi_{\sigma
}\rho^{1/2}\right]  ^{\beta}\right\}  =\operatorname{Tr}\left\{  \left[
\rho^{1/2}\mathcal{N}^{\dag}\left(  \Pi_{\mathcal{N}\left(  \rho\right)
}\right)  \rho^{1/2}\right]  ^{\beta}\right\}  \nonumber\\
&  =\operatorname{Tr}\left\{  \left[  \rho^{1/2}U^{\dag}\left(  \Pi
_{\mathcal{N}\left(  \rho\right)  }\otimes I_{E}\right)  U\rho^{1/2}\right]
^{\beta}\right\}  =\operatorname{Tr}\left\{  \left[  \left(  \Pi
_{\mathcal{N}\left(  \rho\right)  }\otimes I_{E}\right)  U\rho U^{\dag}\left(
\Pi_{\mathcal{N}\left(  \rho\right)  }\otimes I_{E}\right)  \right]  ^{\beta
}\right\}  \nonumber\\
&  =\operatorname{Tr}\left\{  \left[  U\rho U^{\dag}\right]  ^{\beta}\right\}
=\operatorname{Tr}\left\{  \rho^{\beta}\right\}  .
\end{align}
So then%
\begin{equation}
\left.  \frac{d}{d\beta}\widetilde{Q}_{1,\beta}\left(  \rho,\sigma
,\mathcal{N}\right)  \right\vert _{\beta=1}=\left.  \frac{d}{d\beta
}\operatorname{Tr}\left\{  \rho^{\beta}\right\}  \right\vert _{\beta
=1}=\left.  \operatorname{Tr}\left\{  \rho^{\beta}\log\rho\right\}
\right\vert _{\beta=1}=\operatorname{Tr}\left\{  \rho\log\rho\right\}
.\label{eq:Q_beta_1_term}%
\end{equation}
Now we turn to the other term $\frac{d}{d\alpha}\widetilde{Q}_{\alpha
,1}\left(  \rho,\sigma,\mathcal{N}\right)  $. First consider that%
\begin{align}
\frac{d}{d\alpha}\left(  -\alpha^{\prime}\right)   &  =\frac{d}{d\alpha
}\left(  \frac{1-\alpha}{\alpha}\right)  =\frac{d}{d\alpha}\left(  \frac
{1}{\alpha}-1\right)  =-\frac{1}{\alpha^{2}},\\
\widetilde{Q}_{\alpha,1}\left(  \rho,\sigma,\mathcal{N}\right)   &
=\operatorname{Tr}\left\{  \rho\sigma^{-\alpha^{\prime}/2}\mathcal{N}^{\dag
}\left(  \mathcal{N}\left(  \sigma\right)  ^{\alpha^{\prime}/2}\mathcal{N}%
\left(  \rho\right)  ^{-\alpha^{\prime}}\mathcal{N}\left(  \sigma\right)
^{\alpha^{\prime}/2}\right)  \sigma^{-\alpha^{\prime}/2}\right\}  .
\end{align}
Now we show that $\frac{d}{d\alpha}\widetilde{Q}_{\alpha,1}\left(  \rho
,\sigma,\mathcal{N}\right)  $ is equal to%
\begin{multline}
\frac{d}{d\alpha}\operatorname{Tr}\left\{  \rho\sigma^{-\alpha^{\prime}%
/2}\mathcal{N}^{\dag}\left(  \mathcal{N}\left(  \sigma\right)  ^{\alpha
^{\prime}/2}\mathcal{N}\left(  \rho\right)  ^{-\alpha^{\prime}}\mathcal{N}%
\left(  \sigma\right)  ^{\alpha^{\prime}/2}\right)  \sigma^{-\alpha^{\prime
}/2}\right\}  \\
=\operatorname{Tr}\left\{  \rho\left[  \frac{d}{d\alpha}\sigma^{-\alpha
^{\prime}/2}\right]  \mathcal{N}^{\dag}\left(  \mathcal{N}\left(
\sigma\right)  ^{\alpha^{\prime}/2}\mathcal{N}\left(  \rho\right)
^{-\alpha^{\prime}}\mathcal{N}\left(  \sigma\right)  ^{\alpha^{\prime}%
/2}\right)  \sigma^{-\alpha^{\prime}/2}\right\}  \\
+\operatorname{Tr}\left\{  \rho\sigma^{-\alpha^{\prime}/2}\mathcal{N}^{\dag
}\left(  \left[  \frac{d}{d\alpha}\mathcal{N}\left(  \sigma\right)
^{\alpha^{\prime}/2}\right]  \mathcal{N}\left(  \rho\right)  ^{-\alpha
^{\prime}}\mathcal{N}\left(  \sigma\right)  ^{\alpha^{\prime}/2}\right)
\sigma^{-\alpha^{\prime}/2}\right\}  \\
+\operatorname{Tr}\left\{  \rho\sigma^{-\alpha^{\prime}/2}\mathcal{N}^{\dag
}\left(  \mathcal{N}\left(  \sigma\right)  ^{\alpha^{\prime}/2}\left[
\frac{d}{d\alpha}\mathcal{N}\left(  \rho\right)  ^{-\alpha^{\prime}}\right]
\mathcal{N}\left(  \sigma\right)  ^{\alpha^{\prime}/2}\right)  \sigma
^{-\alpha^{\prime}/2}\right\}  \\
+\operatorname{Tr}\left\{  \rho\sigma^{-\alpha^{\prime}/2}\mathcal{N}^{\dag
}\left(  \mathcal{N}\left(  \sigma\right)  ^{\alpha^{\prime}/2}\mathcal{N}%
\left(  \rho\right)  ^{-\alpha^{\prime}}\left[  \frac{d}{d\alpha}%
\mathcal{N}\left(  \sigma\right)  ^{\alpha^{\prime}/2}\right]  \right)
\sigma^{-\alpha^{\prime}/2}\right\}  \\
+\operatorname{Tr}\left\{  \rho\sigma^{-\alpha^{\prime}/2}\mathcal{N}^{\dag
}\left(  \mathcal{N}\left(  \sigma\right)  ^{\alpha^{\prime}/2}\mathcal{N}%
\left(  \rho\right)  ^{-\alpha^{\prime}}\mathcal{N}\left(  \sigma\right)
^{\alpha^{\prime}/2}\right)  \left[  \frac{d}{d\alpha}\sigma^{-\alpha^{\prime
}/2}\right]  \right\}
\end{multline}%
\begin{multline}
=\frac{1}{\alpha^{2}}\Bigg[-\frac{1}{2}\operatorname{Tr}\left\{  \rho\left[
\log\sigma\right]  \sigma^{-\alpha^{\prime}/2}\mathcal{N}^{\dag}\left(
\mathcal{N}\left(  \sigma\right)  ^{\alpha^{\prime}/2}\mathcal{N}\left(
\rho\right)  ^{-\alpha^{\prime}}\mathcal{N}\left(  \sigma\right)
^{\alpha^{\prime}/2}\right)  \sigma^{-\alpha^{\prime}/2}\right\}  \\
+\frac{1}{2}\operatorname{Tr}\left\{  \rho\sigma^{-\alpha^{\prime}%
/2}\mathcal{N}^{\dag}\left(  \left[  \log\mathcal{N}\left(  \sigma\right)
\right]  \mathcal{N}\left(  \sigma\right)  ^{\alpha^{\prime}/2}\mathcal{N}%
\left(  \rho\right)  ^{-\alpha^{\prime}}\mathcal{N}\left(  \sigma\right)
^{\alpha^{\prime}/2}\right)  \sigma^{-\alpha^{\prime}/2}\right\}  \\
-\operatorname{Tr}\left\{  \rho\sigma^{-\alpha^{\prime}/2}\mathcal{N}^{\dag
}\left(  \mathcal{N}\left(  \sigma\right)  ^{\alpha^{\prime}/2}\left[
\log\mathcal{N}\left(  \rho\right)  \right]  \mathcal{N}\left(  \rho\right)
^{-\alpha^{\prime}}\mathcal{N}\left(  \sigma\right)  ^{\alpha^{\prime}%
/2}\right)  \sigma^{-\alpha^{\prime}/2}\right\}  \\
+\frac{1}{2}\operatorname{Tr}\left\{  \rho\sigma^{-\alpha^{\prime}%
/2}\mathcal{N}^{\dag}\left(  \mathcal{N}\left(  \sigma\right)  ^{\alpha
^{\prime}/2}\mathcal{N}\left(  \rho\right)  ^{-\alpha^{\prime}}\mathcal{N}%
\left(  \sigma\right)  ^{\alpha^{\prime}/2}\left[  \log\mathcal{N}\left(
\sigma\right)  \right]  \right)  \sigma^{-\alpha^{\prime}/2}\right\}  \\
-\frac{1}{2}\operatorname{Tr}\left\{  \rho\sigma^{-\alpha^{\prime}%
/2}\mathcal{N}^{\dag}\left(  \mathcal{N}\left(  \sigma\right)  ^{\alpha
^{\prime}/2}\mathcal{N}\left(  \rho\right)  ^{-\alpha^{\prime}}\mathcal{N}%
\left(  \sigma\right)  ^{\alpha^{\prime}/2}\right)  \sigma^{-\alpha^{\prime
}/2}\left[  \log\sigma\right]  \right\}  \Bigg].
\end{multline}
Taking the limit as $\alpha\rightarrow1$ gives%
\begin{multline}
\left.  \frac{d}{d\alpha}\widetilde{Q}_{\alpha,1}\left(  \rho,\sigma
,\mathcal{N}\right)  \right\vert _{\alpha=1}=-\frac{1}{2}\operatorname{Tr}%
\left\{  \rho\left[  \log\sigma\right]  \Pi_{\sigma}\mathcal{N}^{\dag}\left(
\Pi_{\mathcal{N}\left(  \sigma\right)  }\Pi_{\mathcal{N}\left(  \rho\right)
}\Pi_{\mathcal{N}\left(  \sigma\right)  }\right)  \Pi_{\sigma}\right\}  \\
+\frac{1}{2}\operatorname{Tr}\left\{  \rho\Pi_{\sigma}\mathcal{N}^{\dag
}\left(  \left[  \log\mathcal{N}\left(  \sigma\right)  \right]  \Pi
_{\mathcal{N}\left(  \sigma\right)  }\Pi_{\mathcal{N}\left(  \rho\right)  }%
\Pi_{\mathcal{N}\left(  \sigma\right)  }\right)  \Pi_{\sigma}\right\}  \\
-\operatorname{Tr}\left\{  \rho\Pi_{\sigma}\mathcal{N}^{\dag}\left(
\Pi_{\mathcal{N}\left(  \sigma\right)  }\left[  \log\mathcal{N}\left(
\rho\right)  \right]  \Pi_{\mathcal{N}\left(  \rho\right)  }\Pi_{\mathcal{N}%
\left(  \sigma\right)  }\right)  \Pi_{\sigma}\right\}  \\
+\frac{1}{2}\operatorname{Tr}\left\{  \rho\Pi_{\sigma}\mathcal{N}^{\dag
}\left(  \Pi_{\mathcal{N}\left(  \sigma\right)  }\Pi_{\mathcal{N}\left(
\rho\right)  }\Pi_{\mathcal{N}\left(  \sigma\right)  }\left[  \log
\mathcal{N}\left(  \sigma\right)  \right]  \right)  \Pi_{\sigma}\right\}  \\
-\frac{1}{2}\operatorname{Tr}\left\{  \rho\Pi_{\sigma}\mathcal{N}^{\dag
}\left(  \Pi_{\mathcal{N}\left(  \sigma\right)  }\Pi_{\mathcal{N}\left(
\rho\right)  }\Pi_{\mathcal{N}\left(  \sigma\right)  }\right)  \left[
\log\sigma\right]  \Pi_{\sigma}\right\}
\end{multline}
We now simplify the first three terms and note that the last two are Hermitian
conjugates of the first two:%
\begin{align}
\operatorname{Tr}\left\{  \rho\left[  \log\sigma\right]  \Pi_{\sigma
}\mathcal{N}^{\dag}\left(  \Pi_{\mathcal{N}\left(  \sigma\right)  }%
\Pi_{\mathcal{N}\left(  \rho\right)  }\Pi_{\mathcal{N}\left(  \sigma\right)
}\right)  \Pi_{\sigma}\right\}   &  =\operatorname{Tr}\left\{  \rho\left[
\log\sigma\right]  \mathcal{N}^{\dag}\left(  \Pi_{\mathcal{N}\left(
\rho\right)  }\right)  \right\}  \nonumber\\
&  =\operatorname{Tr}\left\{  \mathcal{N}\left(  \rho\left[  \log
\sigma\right]  \right)  \left(  \Pi_{\mathcal{N}\left(  \rho\right)  }\right)
\right\}  \nonumber\\
&  =\operatorname{Tr}\left\{  U\rho\left[  \log\sigma\right]  U^{\dag}\left(
\Pi_{\mathcal{N}\left(  \rho\right)  }\otimes I_{E}\right)  \right\}
\nonumber\\
&  =\operatorname{Tr}\left\{  \Pi_{U\rho U^{\dag}}U\rho U^{\dag}U\left[
\log\sigma\right]  U^{\dag}\left(  \Pi_{\mathcal{N}\left(  \rho\right)
}\otimes I_{E}\right)  \right\}  \nonumber\\
&  =\operatorname{Tr}\left\{  U\rho U^{\dag}U\left[  \log\sigma\right]
U^{\dag}\right\}  \nonumber\\
&  =\operatorname{Tr}\left\{  \rho\left[  \log\sigma\right]  \right\}  ,\\
\operatorname{Tr}\left\{  \rho\Pi_{\sigma}\mathcal{N}^{\dag}\left(  \left[
\log\mathcal{N}\left(  \sigma\right)  \right]  \Pi_{\mathcal{N}\left(
\sigma\right)  }\Pi_{\mathcal{N}\left(  \rho\right)  }\Pi_{\mathcal{N}\left(
\sigma\right)  }\right)  \Pi_{\sigma}\right\}   &  =\operatorname{Tr}\left\{
\rho\mathcal{N}^{\dag}\left(  \left[  \log\mathcal{N}\left(  \sigma\right)
\right]  \Pi_{\mathcal{N}\left(  \rho\right)  }\right)  \right\}  \nonumber\\
&  =\operatorname{Tr}\left\{  \mathcal{N}\left(  \rho\right)  \left[
\log\mathcal{N}\left(  \sigma\right)  \right]  \Pi_{\mathcal{N}\left(
\rho\right)  }\right\}  \nonumber\\
&  =\operatorname{Tr}\left\{  \mathcal{N}\left(  \rho\right)  \left[
\log\mathcal{N}\left(  \sigma\right)  \right]  \right\}  ,\\
\operatorname{Tr}\left\{  \rho\Pi_{\sigma}\mathcal{N}^{\dag}\left(
\Pi_{\mathcal{N}\left(  \sigma\right)  }\left[  \log\mathcal{N}\left(
\rho\right)  \right]  \Pi_{\mathcal{N}\left(  \rho\right)  }\Pi_{\mathcal{N}%
\left(  \sigma\right)  }\right)  \Pi_{\sigma}\right\}   &  =\operatorname{Tr}%
\left\{  \rho\mathcal{N}^{\dag}\left(  \left[  \log\mathcal{N}\left(
\rho\right)  \right]  \Pi_{\mathcal{N}\left(  \rho\right)  }\right)  \right\}
\nonumber\\
&  =\operatorname{Tr}\left\{  \mathcal{N}\left(  \rho\right)  \left(  \left[
\log\mathcal{N}\left(  \rho\right)  \right]  \Pi_{\mathcal{N}\left(
\rho\right)  }\right)  \right\}  \nonumber\\
&  =\operatorname{Tr}\left\{  \mathcal{N}\left(  \rho\right)  \left[
\log\mathcal{N}\left(  \rho\right)  \right]  \right\}  .
\end{align}
This then implies that the following equality holds%
\begin{equation}
\left.  \frac{d}{d\alpha}\widetilde{Q}_{\alpha,1}\left(  \rho,\sigma
,\mathcal{N}\right)  \right\vert _{\alpha=1}=-\operatorname{Tr}\left\{
\mathcal{N}\left(  \rho\left[  \log\sigma\right]  \right)  \right\}
+\operatorname{Tr}\left\{  \mathcal{N}\left(  \rho\right)  \left[
\log\mathcal{N}\left(  \sigma\right)  \right]  \right\}  -\operatorname{Tr}%
\left\{  \mathcal{N}\left(  \rho\right)  \left[  \log\mathcal{N}\left(
\rho\right)  \right]  \right\}  .\label{eq:last-deriv}%
\end{equation}
Putting together (\ref{eq:limit-a-1-exp}), (\ref{eq:chain-break-up}),
(\ref{eq:Q_beta_1_term}), and (\ref{eq:last-deriv}), we can then conclude the
statement of the theorem.
\end{proof}

\section{Petz recovery map}

\label{app:Petz-channel}In this appendix, we explicitly show the well known
fact that the Petz recovery map is a completely-positive,
trace-non-increasing, linear map. Let $\rho$, $\sigma$, and $\mathcal{N}$ be
as in Definition~\ref{def:rho-sig-N}\ and let $\mathcal{R}_{\sigma
,\mathcal{N}}^{P}$\ be the Petz recovery map as defined in
\eqref{eq:Petz-channel-Rel-ent}. Then $\mathcal{R}_{\sigma,\mathcal{N}}^{P}$
is clearly linear, and it is completely positive because it consists of a
composition of the following three completely positive maps:%
\begin{equation}
\left(  \cdot\right)  \rightarrow\left[  \mathcal{N}\left(  \sigma\right)
\right]  ^{-1/2}\left(  \cdot\right)  \left[  \mathcal{N}\left(
\sigma\right)  \right]  ^{-1/2},\ \ \ \left(  \cdot\right)  \rightarrow
\mathcal{N}^{\dag}\left(  \cdot\right)  ,\ \ \ \left(  \cdot\right)
\rightarrow\sigma^{1/2}\left(  \cdot\right)  \sigma^{1/2}.
\end{equation}
It is trace non-increasing because%
\begin{align}
\operatorname{Tr}\left\{  \sigma^{1/2}\mathcal{N}^{\dag}\left(  \left[
\mathcal{N}\left(  \sigma\right)  \right]  ^{-1/2}\left(  X\right)  \left[
\mathcal{N}\left(  \sigma\right)  \right]  ^{-1/2}\right)  \sigma
^{1/2}\right\}   &  =\operatorname{Tr}\left\{  \sigma\mathcal{N}^{\dag}\left(
\left[  \mathcal{N}\left(  \sigma\right)  \right]  ^{-1/2}\left(  X\right)
\left[  \mathcal{N}\left(  \sigma\right)  \right]  ^{-1/2}\right)  \right\}
\nonumber\\
&  =\operatorname{Tr}\left\{  \mathcal{N}\left(  \sigma\right)  \left(
\left[  \mathcal{N}\left(  \sigma\right)  \right]  ^{-1/2}\left(  X\right)
\left[  \mathcal{N}\left(  \sigma\right)  \right]  ^{-1/2}\right)  \right\}
\nonumber\\
&  =\operatorname{Tr}\left\{  \left[  \mathcal{N}\left(  \sigma\right)
\right]  ^{-1/2}\mathcal{N}\left(  \sigma\right)  \left[  \mathcal{N}\left(
\sigma\right)  \right]  ^{-1/2}X\right\}  \nonumber\\
&  =\operatorname{Tr}\left\{  \Pi_{\mathcal{N}\left(  \sigma\right)
}X\right\}  \nonumber\\
&  \leq\operatorname{Tr}\left\{  X\right\}  .
\end{align}

\bibliographystyle{alpha}
\bibliography{Ref}

\end{document}